\documentclass[11pt]{article}
\usepackage{fullpage}
\usepackage{url}
\usepackage{xspace}
\usepackage{color}
\usepackage{graphics}
\usepackage[dvips]{epsfig}

\usepackage{amsmath}
\usepackage{amssymb}
\usepackage{amsfonts}
\usepackage{graphicx}
\usepackage{algorithm2e}

\newtheorem{theorem}{Theorem}[section]
\newtheorem{lemma}{Lemma}[section]

\newtheorem{proposition}{Proposition}[section]

\newtheorem{remark}{Remark}[section]

\newcommand{\qed}{\hfill $\Box$ \bigbreak}
\newenvironment{proof}{\noindent {\bf Proof.}}{\qed}

\newcommand{\cA}{{\cal A}}

\newcommand{\cF}{{\cal F}}
\newcommand{\cE}{{\cal E}}

\newcommand{\remove}[1]{}

\begin{document}

\baselineskip  0.2in 
\parskip     0.1in 
\parindent   0.0in 

\title{{{\bf Rendezvous in Networks in Spite of Delay Faults}}\footnote
{A preliminary version of this paper, entitled ``Fault-Tolerant Rendezvous in Networks", appeared in 
Proc. 41st International Colloquium on Automata, Languages and Programming (ICALP 2014), July 2014, Copenhagen, Denmark, 411-422. Partially supported by NSERC discovery grant 8136 --2013, by the Research Chair in Distributed Computing at the
Universit\'e du Qu\'ebec en Outaouais and by the french ANR project MACARON (anr-13-js02-0002). }}
\date{}
\newcommand{\inst}[1]{$^{#1}$}

\author{
J\'{e}r\'{e}mie Chalopin\inst{1},
Yoann Dieudonn\'e\inst{2},
Arnaud Labourel\inst{1},
Andrzej Pelc\inst{3}\\
\inst{1} LIF, CNRS \& Aix-Marseille University, Marseille, France.\\
E-mails: \url{{jeremie.chalopin,arnaud.labourel}@lif.univ-mrs.fr}\\
\inst{2} Universit\'{e} de Picardie Jules Verne, Amiens, France.\\
E-mail: {\url{yoann.dieudonne@u-picardie.fr}}\\
\inst{3} Universit\'{e} du Qu\'{e}bec en Outaouais, Gatineau, Canada.\\
E-mail: \url{pelc@uqo.ca}\\
}

\date{ }
\maketitle

\begin{abstract}

Two mobile agents, starting from different nodes of an unknown network, have to {meet at a node}.
Agents move in synchronous rounds using a deterministic algorithm. Each agent has a different label, which it can
use in the execution of the algorithm, but it does not know the label of the other agent. 
Agents do not know any bound on the size of the network.
In each round an agent decides if it remains
idle or if it wants to move to one of the adjacent nodes. Agents are subject to {\em delay faults}: if an agent incurs a fault in a given round,
it remains in the current node, regardless of its decision. If it planned to move and the fault happened, the agent is aware of it. We consider three
scenarios of fault distribution: random (independently in each round and for each agent with constant probability $0<p<1$), unbounded 
adversarial (the adversary can delay an agent for an arbitrary finite number of consecutive rounds) and bounded adversarial
(the adversary can delay an agent for at most $c$ consecutive rounds, where $c$ is unknown to the agents). The quality measure of a rendezvous algorithm is its cost, which is the total number of edge traversals.

For random faults, we show an algorithm with cost polynomial in the size $n$ of the network and {\em polylogarithmic} in the larger label $L$, which achieves rendezvous with very high probability in arbitrary networks.
By contrast, for unbounded adversarial faults we show that rendezvous is not {possible}, even in the class of rings.
Under this scenario we give a rendezvous algorithm with cost $O(n\ell)$, where $\ell$ is the smaller label,
working in arbitrary trees, and we show that $\Omega(\ell)$ is the lower bound on rendezvous cost, even for the two-node tree.
For bounded adversarial faults, we give a rendezvous algorithm working  for arbitrary networks, with cost polynomial in $n$,
and {\em logarithmic} in the bound $c$ and in the larger
label $L$. 

\vspace{2ex}


\vspace{2ex}
\begin{center}
\end{center}
\end{abstract}

\vfill

\vfill

\thispagestyle{empty}
\setcounter{page}{0}
\pagebreak

\section{Introduction}

\subsection{The background}

Two mobile entities, called agents, starting from different nodes of a network, have to {meet at a node}.
This task is known  as {\em rendezvous} and has been extensively studied in the literature.
Mobile entities may represent software agents in computer networks, mobile robots, if the network is composed of corridors in a mine,
or people who want to meet in an unknown city whose streets form a network. 
The reason to meet may be to exchange data previously collected by the agents,
or to coordinate a future network maintenance task. In this paper we study a fault-tolerant version of the rendezvous problem:
agents have to meet in spite of delay faults that they can incur during navigation. Such faults may be due to mechanical reasons in the case
of robots and to network congestion in the case of software agents. 

\subsection{The model and the problem}

The network is modeled as an undirected connected graph.
We seek deterministic rendezvous algorithms that do not
rely on the knowledge of node identifiers, and can work in anonymous graphs as well  (cf. \cite{alpern02b}). 
The importance of designing such algorithms
is motivated by the fact that, even when nodes are equipped with distinct identifiers, agents may be unable to perceive them
because of limited sensory capabilities (a robot may be unable to read signs at corridor crossings), 
or nodes may refuse to reveal their identifiers to software agents, e.g., due to security or privacy reasons.
Note that, if nodes had distinct identifiers visible to the agents, the agents might explore the graph and meet at the node
with smallest identifier, hence rendezvous
would reduce to graph exploration.
On the other hand, we assume that
edges incident to a node $v$ have distinct labels (visible to the agents) in 
$\{0,\dots,d-1\}$, where $d$ is the degree of $v$. Thus every undirected
edge $\{u,v\}$ has two labels, which are called its {\em port numbers} at $u$
and at $v$. Port numbering is {\em local}, i.e., there is no relation between
port numbers at $u$ and at $v$. Note that in the absence of port numbers, edges incident to a node
would be {indistinguishable} for agents and thus rendezvous would be often impossible, 
as the adversary could prevent an agent from taking some edge incident to the current node.
Security and privacy reasons for not revealing node identifiers to software agents are irrelevant in the case of port numbers, and 
port numbers in the case of a mine or labyrinth can be made implicit, e.g., by marking one edge at each crossing
(using a simple mark legible by the robot),
considering it as corresponding to port 0 and all other port numbers increasing clockwise.

Agents start at different nodes of the graph and  traverse its edges in synchronous rounds.
They cannot mark visited nodes or traversed edges in any way.
{Initially, the agents are dormant (i.e., they do not execute the algorithm)}. The adversary wakes up each of the agents in possibly different rounds.
Each agent starts executing the algorithm in the round of its wake-up.
It has a clock measuring rounds that starts at its wake-up round.
In each round an agent decides if it remains
idle or if it chooses a port to move to one of the adjacent nodes. Agents are subject to {\em delay faults} in rounds in which they
decide to move: if an agent incurs a fault in such a round,
it remains at the current node and is aware of the fault. We consider three
scenarios of fault distribution: random (independently in each round and for each agent with constant probability $0<p<1$), unbounded 
adversarial (the adversary can delay an agent for an arbitrary finite number of consecutive rounds) and bounded adversarial
(the adversary can delay an agent for at most $c$ consecutive rounds, where $c$ is unknown to the agents).
Agents do not  know the topology of the graph or any bound on its size.
Each agent has a different positive integer label which it knows and can use in the execution of the rendezvous algorithm, but it does not know the label
of the other agent nor its starting round. 
When an agent enters a node, it learns its degree and the port of entry. When agents cross each other
on an edge, traversing it simultaneously in different directions, they do not notice this fact.
We assume that the memory of the agents is unlimited: from the computational point of view they are modeled as 
Turing machines. 

The quality measure of a rendezvous algorithm is its {\em cost}, which is the total number of edge traversals. For each of the considered fault
distributions we are interested in deterministic algorithms working at low cost. For both scenarios with adversarial faults we say that
a deterministic rendezvous algorithm works at a cost at most $C$ for a given class of graphs if for any initial positions in a graph of this class both agents
meet after at most $C$ traversals, regardless of the faults imposed by  the adversary obeying the given scenario.  
In the case of random faults the algorithm is also deterministic, but, due to the stochastic nature of faults, the estimate of its cost is with high probability.

\subsection{Our results}
\label{subsec:ourresults}

For random faults, we show an algorithm which achieves rendezvous  in arbitrary networks at  cost polynomial in the size $n$ of the network and {\em polylogarithmic} in the larger label $L$, with very high probability. More precisely, our algorithm achieves rendezvous with probability 1,  and its cost exceeds a polynomial in $n$ and $\log L$ with probability inverse exponential in $n$ and $\log L$.
By contrast, for unbounded adversarial faults, we show that rendezvous is not feasible, even in the class of rings.
Under this scenario we give a rendezvous algorithm with cost $O(n\ell)$, where $\ell$ is the smaller label,
working in arbitrary trees, and we show that $\Omega(\ell)$ is the lower bound on rendezvous cost, even for the two-node tree.
For bounded adversarial faults we give a rendezvous algorithm working  for arbitrary networks, with cost polynomial in $n$,
and {\em logarithmic} in the bound $c$ and in the larger
label $L$.


\subsection{Related work}
\label{subsec:relatwork}

The problem of rendezvous has been studied both under the randomized and the deterministic scenarios.
An extensive survey of  randomized rendezvous in various models  can be found in
\cite{alpern02b}, cf. also  \cite{alpern95a,alpern02a,anderson90,baston98,israeli}. 
Deterministic rendezvous in networks has been surveyed in \cite{Pe}.
Several authors
considered the geometric scenario (rendezvous in an interval of the real line, see, e.g.,  \cite{baston98,baston01},
or in the plane, see, e.g., \cite{anderson98a,anderson98b}).
Gathering more than two agents has been studied, e.g., in \cite{fpsw,israeli,lim96,thomas92}.

For the deterministic setting many authors studied the feasibility of synchronous rendezvous, and the time required to achieve this task, when feasible. For instance, deterministic rendezvous of agents equipped with tokens used to mark nodes was considered, e.g., in~\cite{KKSS}. Deterministic rendezvous of two agents that cannot mark nodes but have unique labels was discussed in \cite{DFKP,KM,TSZ07}.
Since this is our scenario, these papers are the most relevant in our context. All of them are concerned with the time of rendezvous in arbitrary
graphs. In \cite{DFKP} the authors show a rendezvous algorithm polynomial in the size of the graph, in the length of the shorter
label and in the delay between the starting times of the agents. In \cite{KM,TSZ07} rendezvous time is polynomial in the first two of these parameters and independent of the delay.

Memory required by the agents to achieve deterministic rendezvous has been studied in \cite{FP2} for trees and in  \cite{CKP} for general graphs.
Memory needed for randomized rendezvous in the ring is discussed, e.g., in~\cite{KKPM08}. 

Apart from the synchronous model used in this paper, several authors investigated asynchronous rendezvous in the plane \cite{CFPS,fpsw} and in network environments
\cite{BCGIL,CLP,DGKKP,DPV}.
In the latter scenario the agent chooses the edge which it decides to traverse but the adversary controls the speed of the agent. Under this assumption rendezvous
in a node cannot be guaranteed even in very simple graphs and hence the rendezvous requirement is relaxed to permit the agents to meet inside an edge.

Fault-tolerant aspects of the rendezvous problem have been investigated in \cite{CDW,Da,DMSVW,DPP,FKKLSS}.
Faulty unmovable tokens were considered in the context of the task of gathering many agents at one node.
In \cite{Da,FKKLSS} the authors considered gathering in rings, and in \cite{DMSVW} gathering was studied
in arbitrary graphs, under the assumption that an unmovable token is located in the starting node of each agent.
Tokens could disappear during the execution of the algorithm, but they could not reappear again. Byzantine tokens
which can appear and disappear arbitrarily  have been considered in \cite {DP} for the related task of network exploration.
A different fault scenario for gathering many agents was investigated in \cite{DPP}. The authors assumed that some number of agents
are Byzantine and they studied the problem of how many good agents are needed to guarantee meeting of all of them despite the actions
of Byzantine agents. To the best of our knowledge rendezvous with delay faults considered in the present paper has never been studied before.

\section{Preliminaries}\label{prelim}

Throughout the paper, 
the number of nodes of a graph is called its size.
In this section we recall two procedures known from the literature, that will be used as building blocks in some of our algorithms. 
The aim of the first procedure is graph exploration, i.e., visiting all nodes and traversing all edges of the graph by a single agent. 
The procedure, based on universal exploration sequences (UXS) \cite{Ko}, is a corollary of the  result of Reingold \cite{Re}. Given any positive integer $m$, it allows the agent to traverse all edges of any graph of size at most $m$,
starting from any node of this graph, using $P(m)$ edge traversals, where $P$ is some polynomial. (The original procedure of Reingold only visits all nodes, but it can be transformed to traverse all edges by visiting all neighbors of each visited node before going to the next node.) After entering a node of degree $d$ by some port $p$,
the agent can compute the port $q$ by which it has to exit; more precisely $q=(p+x_i)\mod d$, where $x_i$ is the corresponding term of the UXS.

A {\em trajectory} is a sequence of nodes of a graph, in which each node is adjacent to the preceding one. 
Given any starting node $v$,  we denote by $R(m,v)$ the trajectory obtained by Reingold's procedure
followed by its reverse. (Hence the trajectory starts and ends at node $v$.)  The procedure can be applied in any graph starting at any node, giving
some trajectory. We say that  the agent {\em follows} a trajectory if it executes the above procedure used to construct it.
This trajectory will be called {\em integral}, if the corresponding route covers all edges of the graph. By definition, the trajectory $R(m,v)$ is integral if it is
obtained by Reingold's procedure applied in any graph of size at most $m$ starting at any node~$v$. 

The second auxiliary procedure is the Algorithm {\tt RV-asynch-poly} from \cite{DPV} that guarantees rendezvous of two agents under the {\em asynchronous}
scenario.
Unlike in the synchronous scenario used in the present paper, in the asynchronous scenario each agent chooses consecutive ports that it wants to use but the adversary controls the speed of the agent, changing it arbitrarily during navigation. Rendezvous is guaranteed in the asynchronous scenario, if it occurs for any
behavior of the adversary.
Under this assumption rendezvous
in a node cannot be guaranteed even in very simple graphs and hence the rendezvous requirement is relaxed to permit the agents to meet inside an edge.
Recall that in our synchronous scenario, agents crossing each other on an edge traversing it simultaneously in different directions, not only do not meet but do not even notice the fact of crossing.

Algorithm {\tt RV-asynch-poly} works at cost polynomial in the size $n$ of the graph in which the agents operate and in the length of the smaller label. 
Let $A$ be a polynomial, such that if two agents with different labels $\lambda_1$ and $\lambda_2$ execute 
Algorithm {\tt RV-asynch-poly} in an $n$-node graph, then the agents meet in the asynchronous model, after at most 
$A(n,\min(\log \lambda_1, \log \lambda_2))$ steps.

\section{Random faults}

In this section we consider the scenario when agents are subject to random and independent faults. More precisely, for each agent and each round the probability that
the agent is delayed in this round is $0<p<1$, where $p$ is a constant, and the events of delaying are independent for each round and each agent.
Under this scenario we construct a deterministic rendezvous algorithm that achieves rendezvous in any connected graph with probability 1 and its cost exceeds a polynomial in $n$ and $\log L$ with probability inverse exponential in $n$ and $\log L$, where $n$ is the size of the graph and $L$ is the larger label.

The intuition behind the algorithm is the following. Since the occurrence of random faults represents a possible behavior of the asynchronous adversary
in Algorithm {\tt RV-asynch-poly} from \cite{DPV}, an idea to get the guarantee of a meeting with random faults at polynomial cost might be to only use this algorithm. However, this meeting 
may occur either at a node or inside an edge, according to the model from \cite{DPV}. In the synchronous model with random faults
considered in this section, the second type of meeting is not considered as rendezvous, in fact agents do not even notice it. Hence we must construct a  {\em deterministic} mechanism which guarantees a legitimate meeting at a node, with high probability, soon after an ``illegitimate'' meeting inside an edge. Constructing this mechanism and proving that it works as desired is the main challenge of rendezvous with random faults. 

\subsection{The algorithm}

Before describing the algorithm we define the following transformation of the label $\lambda$ of an agent. Let $\Phi(0)=(0011)$ and $\Phi(1)=(1100)$.
Let $(c_1\dots c_k)$ be the binary representation of the label $\lambda$.  We define the {\em modified label} $\lambda^*$  of the agent as the concatenation
of sequences $\Phi(c_1),\dots, \Phi(c_k)$ and $(10)$. Note that {if the labels} of two agents are different, then their transformed labels are different and none of them is a prefix of the other.

We first describe the procedure {\tt Dance} $(\lambda, x,y)$ executed by an agent with label $\lambda$ located at node $y$ at the start of the procedure.
Node $x$ is a node adjacent to $y$.

\noindent
{\bf Procedure} {\tt Dance} $(\lambda, x,y)$ 

\noindent
Let $\lambda^*=(b_1,\dots,b_m)$.

\noindent
Stage 1.\\
Stay idle at $y$ for 10 rounds.

\noindent
Stage 2.\\
{\bf for} $i=1$ {\bf to} $m$ {\bf do}\\
\hspace*{1cm}{\bf if } $b_i=0$\\
\hspace*{1cm}{\bf then} stay idle for two rounds\\
\hspace*{1cm}{\bf else} go to $x$ and in the next round return to $y$.

\noindent
Stage 3.\\Traverse the edge $\{x,y\}$ 12 times {(i.e., go back and forth 6 times across the edge $\{x,y\}$)}.
\hfill $\diamond$

Note that procedure {\tt Dance} $(\lambda, x,y)$ has cost $O(\log \lambda)$.

We will also use procedure {\tt Asynch}$(\lambda)$ executed by an agent with label $\lambda$ starting at any node $x_0$ of a graph. This procedure produces
an infinite walk $(x_0,x_1,x_2,\dots)$ in the graph resulting from applying Algorithm {\tt RV-asynch-poly} by a single agent with label $\lambda$.

Using these procedures we now describe Algorithm {\tt RV-RF} (for rendezvous with random faults), that works for an agent with label $\lambda$
starting at an arbitrary node of any connected graph. 

\noindent
{\bf Algorithm} {\tt RV-RF}

The algorithm works in two phases interleaved in a way depending on faults occurring in the execution and repeated until rendezvous.
The agent starts executing the algorithm in phase Progress.

\noindent
{\em Phase Progress} 

This phase proceeds in stages. Let $(x_0,x_1,x_2,\dots)$ be the infinite walk produced by the agent starting at node $x_0$ and applying 
{\tt Asynch}$(\lambda)$. The $i$th stage of phase Progress, for $i \geq 1$,  is the traversal of the edge $\{x_{i-1},x_i\}$ from $x_{i-1}$ to $x_i$, followed by
the execution of  {\tt Dance} $(\lambda, x_{i-1},x_i)$. The agent executes consecutive stages of phase Progress until a fault occurs.

If a fault occurs in the first round of the $i$th stage, then the agent repeats the attempt of this traversal again, until success
and then continues with {\tt Dance} $(\lambda, x_{i-1},x_i)$. If a fault occurs in the $t$th round  of the $i$th stage, for $t>1$, i.e.,  during the execution of 
procedure  {\tt Dance} $(\lambda, x_{i-1},x_i)$ in the $i$th stage, then this execution is interrupted and phase Correction is launched starting at the node
where the agent was situated when the fault occurred.

\noindent
{\em Phase Correction} 

Let $e$ denote the edge $\{x_{i-1},x_i\}$ and let $w$ be the node at which the agent was situated when the last fault occurred during the execution of 
{\tt Dance} $(\lambda, x_{i-1},x_i)$. Hence $w$ is either $x_{i-1}$ or $x_i$.

\noindent
Stage 1.\\
Stay idle at $w$ for 20 rounds.

\noindent
Stage 2.\\
Traverse edge $e$ 20 times.

\noindent
Stage 3.\\
If the agent is not at $w$, then go to $w$. 

If a fault occurs during the execution of phase Correction, then the execution of this phase is dropped and a new phase Correction is launched from the beginning,
starting at the node where the agent was situated when the fault occurred. Upon completing an execution of the phase Correction without any fault the agent is at node $w$. It resumes the execution of the $t$th round of the $i$th stage of phase Progress.
\hfill $\diamond$

\subsection{Correctness and analysis}

This section is devoted to the proof of correctness and analysis of performance of Algorithm {\tt RV-RF}. It is split into a series of lemmas.
The first lemma is straightforward.

\begin{lemma}\label{10}
In every segment of 10 consecutive rounds {without any faults} of execution of Stage 2 of procedure {\tt Dance} the agent is idle at least once
and moves at least once.
\end{lemma}

\begin{lemma}\label{20}
In every segment of 20 consecutive rounds {without any faults} of execution of procedure {\tt Dance} the agent moves at least once.
\end{lemma}

\begin{proof}
Since Stage 1 of procedure {\tt Dance} consists of 10 rounds in which the agent is idle and in Stage 3 the agent is never idle,
the lemma follows from the second part of Lemma \ref{10}.
 
\end{proof}

In our reasoning we will consider an auxiliary model $\cal M$ of the behavior of agents and of the type of faults. Agents move in synchronous rounds of constant duration T. An edge traversal is always performed at a constant speed and so that the destination node is reached exactly at the end of the round involving the traversal. The agents do not notice when they meet, neither at a node nor inside an edge. Hence, when they execute Algorithm {\tt RV-RF},
they do so indefinitely in an independent way.  Faults in model $\cal M$ are 
unbounded adversarial, i.e., the adversary can delay an agent at a node for an arbitrary finite number of consecutive rounds.
When an agent executes procedure {\tt Asynch}$(\lambda)$ in model $\cal M$, it attempts to make the next step of the procedure in each round, but can be 
delayed by the adversary at each step. 
For rounds $t<t'$ we denote by $[t,t']$ the time interval between the beginning of round
$t$ and the end of round $t'$. We use $(t,t')$ instead of $[t+1,t'-1]$.
For convenience, we will sometimes use the phrase ``in round $t$'' instead of ``at the end of round $t$''
and ``by round $t$'' instead of ``by the end of round $t$''.
Considerations in this auxiliary model will serve us to draw conclusions about rendezvous in the random fault model.

A meeting in model $\cal M$ is defined as both agents being at the same node at the same time or being in the same point inside an edge at the same time.  We use the word {\em meeting} in the auxiliary model $\cal M$  to differentiate it from {\em rendezvous} in our principal model: the first may occur at a node or inside an edge and agents do not notice it, and the second can occur only at a node, agents notice it and stop. 
Notice that five types of meetings are possible in model $\cal M$.

Type 1. The agents cross each other inside an edge $\{u,v\}$ in some round, one agent going from $u$ to $v$ and the other going from $v$ to $u$.

Type 2. The agents stay together inside an edge during its traversal in the same round in the same direction.

Type 3. The agents meet at node $v$ coming from the same node $u$, not necessarily in the same round.

Type 4. The agents meet at node $v$ coming from different nodes $u$ and $w$, not necessarily in the same round.

Type 5. The agents meet at node $v$, such that  one of them has never moved from $v$.



Hence when there is a meeting in model $\cal M$, only one of the following 9 situations can occur:

Situation A1. The agents cross each other inside an edge $\{u,v\}$ in some round, one agent going from $u$ to $v$ and the other going from $v$ to $u$.

Situation A2. Agents meet at a node $v$, such that one of them has never moved yet.

Situation A3. Agents meet at a node $v$, both coming from the same node $u$ in different rounds. 

Situation A4. Agents $a$ and b meet at a node $v$, such that:
\begin{enumerate}
\item
agent $a$ comes from a node $u$ to $v$ {in its $r$th edge traversal} in round $k_{1,a}$ and
after the meeting goes to a node {$w\ne u$} {in its $(r+1)$th edge traversal} in round $k_{2,a}$;
\item
agent $b$ comes from node $w$ to $v$ {in its $s$th edge traversal} in round $k_{1,b}$ and after the meeting goes
to node $u$ {in its $(s+1)$th edge traversal} in round $k_{2,b}$.
\item
{rounds $k_{1,a}$ and $k_{1,b}$ (resp. $k_{2,a}$ and $k_{2,b}$) are not necessarily the same and $max(k_{1,a},k_{1,b})<min(k_{2,a},k_{2,b})$.}
\end{enumerate}

Situation A5. Agents $a$ and $b$ meet at a node $v$, such that:
\begin{enumerate}
\item
agent $a$ comes to $v$ from $u$ in its $r$th edge traversal {in round $k_{1,a}$} and {after the meeting} goes to node {$w\ne u$} in its $(r+1)$th edge traversal {in round $k_{2,a}$}; 
\item
agent $b$ comes to $v$ from $w$ in its $s$th edge traversal {in round $k_{1,b}$} and {after the meeting} goes to a node $z \neq u$ in its $(s+1)$th edge traversal {in round $k_{2,b}$};
\item
{rounds $k_{1,a}$ and $k_{1,b}$ are not necessarily the same and $max(k_{1,a},k_{1,b})<k_{2,a}\leq k_{2,b}$.}
\end{enumerate}

Situation B1.
Agents meet at a node $v$, both coming from the same node $u$ in the same round. 

Situation B2
Agents meet inside an edge $\{u,v\}$, both coming from $u$ and going to $v$ in the same round.
Since they move at the same constant speed, they
are simultaneously at each point inside the edge $\{u,v\}$ during the traversal. We will consider this as a single meeting.

Situation B3.
Agents meet at a node $v$, such that:
\begin{enumerate}
\item
agent $a$ comes from a node $u$ to $v$ {in its $r$th edge traversal} in round $k_{1,a}$  and after the
meeting goes to a node $w$ {in its $(r+1)$th edge traversal} in round $k_{2,a}$;
\item
 agent $b$ comes from node $p$ to $v$ {in its $s$th edge traversal} in round $k_{1,b}$ and after the meeting goes to node
$q$ {in its $(s+1)$th edge traversal} in round $k_{2,b}$,
\item
{$p \ne w$,  $q\ne u$;  rounds $k_{1,a}$ and $k_{1,b}$ (resp. $k_{2,a}$ and $k_{2,b}$) are not necessarily the same; $max(k_{1,a},k_{1,b})<min(k_{2,a},k_{2,b})$.}
\end{enumerate}

Situation B4.
Agents meet at node $v$ such that
\begin{enumerate}
\item
agent $a$ comes to $v$ from $u$ in its $r$th edge traversal {in round $k_{1,a}$} and {after the meeting} goes to node {$w\ne u$} in its $(r+1)$th edge traversal {in round $k_{2,a}$}; 
\item
agent $b$ comes to $v$ from $w$ in its $s$th edge traversal {in round $k_{1,b}$} and {after the meeting} goes to a node $z \neq u$ in its $(s+1)$th edge traversal {in round $k_{2,b}$};
\item
{rounds $k_{1,a}$ and $k_{1,b}$ are not necessarily the same and $max(k_{1,a},k_{1,b})<k_{2,b}<k_{2,a}$.}

\end{enumerate}

The next lemma shows that one of the situations A1 -- A5 is unavoidable when applying Asynch in the model $\cal M$.

\begin{lemma}\label{5 situations}
Consider two agents $a$ and $b$, with different labels $\lambda_1$ and $\lambda_2$, respectively,  starting at arbitrary different nodes of an $n$-node graph, where $n$ is unknown to the agents. In the model $\cal M$, if agent $a$ applies procedure {\tt Asynch}( $\lambda_1$) and
agent $b$ applies procedure {\tt Asynch}( $\lambda_2$), then, for every behavior of the adversary at least one of the situations
A1 -- A5 must occur after a total of at most $A(n,\min(\log \lambda_1, \log \lambda_2))$ steps of both agents.
\end{lemma}

\begin{proof}
Suppose, for contradiction, that there exists a behavior of the adversary in model $\cal M$ such that all situations A1 -- A5 can be
avoided during the first $A(n,\min(\log \lambda_1, \log \lambda_2))$ steps of both agents. 
Denote by $S$ the scenario truncated to the first $A(n,\min(\log \lambda_1, \log \lambda_2))$ steps of both agents resulting from the above behavior.
Since Algorithm {\tt RV-asynch-poly} guarantees a meeting under any behavior of an adversary in the asynchronous model after a total of at most $A(n,\min(\log \lambda_1, \log \lambda_2))$ steps of both agents (cf. \cite{DPV}), 
this means that in scenario $S$, which takes place in the model $\cal M$,  at least one of the situations B1 -- B4 must occur.
We will show that based on scenario $S$ in model $\cal M$ it is possible to construct  a scenario $AS$ in the asynchronous model
from \cite{DPV} with no meetings after a total of at most $A(n,\min(\log \lambda_1, \log \lambda_2))$ steps,
which contradicts the result from  \cite{DPV}. 

Consider the first meeting $\rho$ in scenario $S$. This meeting cannot be in situation B1 or B2
because then the agents would meet previously at node $u$. Hence either situation B3 or B4 must occur. 

First suppose that 
situation B4 occurred during the meeting $\rho$. Split the scenario $S$ into two parts $S_1$ and $S_2$,  such that $S_2$ 
is the part of scenario $S$ that consists of all rounds after the round in which agent $b$ makes its $(s+1)$th edge traversal.
In particular, in the beginning of $S_2$ agent $a$ has already made its $r$th edge traversal but not yet its $(r+1)$th edge traversal.
$S_1$ is the part of scenario $S$ preceding $S_2$. Notice that $S_1$ contains only one meeting, the meeting $\rho$.

We can modify the behavior of the adversary in scenario $S_1$ to produce scenario $S'_1$ in which there is no meeting.
In order to do so, consider 3 cases.

Case 1. The meeting $\rho$ occurs in the round when both $a$ and $b$ arrive at $v$.

In the round preceding the execution of its $r$th edge traversal agent $a$ is at node $u$, while $b$ is at node $w \neq u$.
According to scenario $S_1$, in the next round both agents go to $v$. The modification to obtain scenario $S'_1$ is as follows. The adversary delays
agent $a$ at $u$ for one round and releases agent $b$ to make its $s$th edge traversal to $v$. Then the adversary releases agent $a$ to go to $v$ and agent $b$ to make its $(s+1)$th edge traversal to $z$.
This avoids the meeting.

Case 2. The meeting $\rho$ occurs in a round in which $b$ is idle at $v$ (it came there in a previous round) and $a$ comes to $v$ from $u$ executing its $r$th edge traversal.

In the round preceding this traversal agent $a$ was at $u$. The modification to obtain scenario $S'_1$ is as follows. 
The adversary releases $b$ to make its $(s+1)$th edge traversal to $z$
and in the same round releases $a$ to make its $r$th edge traversal.
This avoids the meeting.

Case 3.
 The meeting $\rho$ occurs in a round in which $a$ is idle at $v$ (it came there in a previous round) and $b$ comes to $v$ from $w$ executing its $s$th edge traversal.

Let $t$ be the round in which agent $a$ makes its $r$th edge traversal and $t'>t$ the round in which agent $b$ makes its $s$th edge traversal.
In the time interval $[x,y]$, where $x$ is the beginning of round $t$ and $y$ is the end of round $t'-1$, agent $b$ does not traverse edge $\{u,v\}$ because otherwise the agents would previously meet inside this edge
or at $v$. 
The modification to obtain scenario $S'_1$ is as follows. The adversary starts moving 
agent $a$ from $u$ to $v$ in the beginning of round $t$, but blocks it inside the edge $\{u,v\}$ until the end of the round $t'$ in which
$b$ makes its $s$th edge traversal to $v$, which it is released to do. In the next round the adversary releases $b$ to make its
$(s+1)$th edge traversal to $z$ and releases $a$ to finish its $r$th edge traversal to $v$ (these two actions finish simultaneously). 
This avoids the meeting.

Hence in all cases we obtain a scenario $S'_1$ which does not contain any meeting. 
Notice that scenario $S'_1$ is still a scenario in model $\cal M$ in cases 1 and 2, but is {\em not} a scenario in this model in case  3.
Indeed, in this case agent $a$ does not travel with constant speed inside the edge $\{u,v\}$. However, in all cases this is a legitimate scenario for the
asynchronous adversary. It ends at the end of  a round when
$a$ has made its $r$th edge traversal and when $b$ has made its $(s+1)$th edge traversal but {before $b$ has started} its $(s+2)$th edge traversal.
Hence scenario $S'$ consisting of scenario $S'_1$ followed by $S_2$ is possible for the asynchronous adversary, has one fewer meeting than scenario $S$ and all situations A1 -- A5 are still avoided.

In a similar way it can be shown that if situation B3 occurred during the meeting $\rho$, then a scenario avoiding the first meeting
of $S$ can be produced. Notice that scenario $S_2$ remained unchanged and it is still a scenario in model $\cal M$. The same reasoning can be now applied to scenario $S_2$, again
transforming it into a scenario with one fewer meeting (avoiding the first meeting in $S_2$), the last part of which (containing other meetings, if any, except the first 
meeting of scenario $S_2$) is still a scenario in
model $\cal M$.

By induction on the number of meetings it follows that a scenario $AS$ for some behavior of the asynchronous adversary
(not for model $\cal M$ any more) without any meeting after a total of $A(n,\min(\log \lambda_1, \log \lambda_2))$ steps can be produced. 
This, however,
contradicts the fact that Algorithm {\tt RV-asynch-poly} guarantees a meeting
under any behavior of the asynchronous adversary after a total of at most $A(n,\min(\log \lambda_1, \log \lambda_2))$ steps of both agents.
\end{proof}

\begin{lemma}\label{5 events}
Consider two agents $a$ and $b$, with different labels $\lambda_1$ and $\lambda_2$, respectively,  starting at arbitrary different nodes of an $n$-node graph, where $n$ is unknown to the agents. In the model $\cal M$, if  the agents have executed at least $A(n,\min(\log \lambda_1, \log \lambda_2))$
stages  of phase Progress of Algorithm {\tt RV-RF}, then at least one of the following events must have occurred:

Event 1. There exists a round $k$ during which the agents cross each other inside an edge $\{u,v\}$, one agent going from $u$ to $v$ and the other going from $v$ to $u$, each of the agents applying a step of procedure {\tt Asynch}.

Event 2.  There exists a round $k$ during which one agent arrives at node $v$ applying a step of procedure {\tt Asynch} and
the other agent has never moved yet.

Event 3. There exists a round $k$ during which one agent arrives at node $v$ from $u$ applying a step of procedure {\tt Asynch}, and the other agent has also arrived at $v$ from $u$ applying a step of procedure {\tt Asynch} in some round $k'<k$ and has not made
any further step of procedure {\tt Asynch} until the end of round $k$.

Event 4. There exists a round $k$ during which agent $b$ arrives at node $v$ from $w$ applying a step of procedure {\tt Asynch}
and such that:
\begin{itemize}
\item
agent $a$  has arrived at $v$ from {$u\ne w$} applying a step of procedure {\tt Asynch} in some round $k'\leq k$ and has not made
any further step of procedure {\tt Asynch} during the time interval $[k',k]$;
\item
the next step of procedure {\tt Asynch} after round $k$ brings agent $a$ to node $w$;
\item
the next step of procedure {\tt Asynch} after round $k$ brings agent $b$ to node $u$.
\end{itemize}

Event 5. There exists a round $k_1$ during which agent $a$ arrives at node $v$ from $u$ applying a step of procedure {\tt Asynch}
and such that:
\begin{itemize}
\item
the next step of procedure {\tt Asynch} after round $k_1$ performed by agent $a$ in a round $k_2>k_1$, brings agent $a$ to node {$w\ne u$};
\item
agent $b$ arrives at $v$ from $w$ applying a step of procedure {\tt Asynch} in some round $p_1<k_2$ and goes from
$v$ to some node $z\neq u$ applying the next step of procedure {\tt Asynch} in some round  $ p_2\geq k_2$. 
\end{itemize}
\end{lemma}

\begin{proof}
Consider a scenario $S$ in model $\cal M$ in which none of the events 1 -- 5 occurred
by the time the agents completed  $A(n,\min(\log \lambda_1, \log \lambda_2))$ stages of phase Progress  of Algorithm {\tt RV-RF}. Replace in scenario $S$ each action
that does not correspond to a step of procedure {\tt Asynch} by (a tentative of) performing by the agent the next step of procedure {\tt Asynch} and 
imposing in this round a delay fault by the adversary. The obtained scenario $S'$ is a legitimate scenario in the execution of procedure {\tt Asynch} in model
$\cal M$ by agents with labels $\lambda_1$ and $\lambda_2$ in which, after $A(n,\min(\log \lambda_1, \log \lambda_2))$ steps of both agents executed in this procedure none of the situations A1 -- A5 took place. This contradicts Lemma~
\ref{5 situations}.

\end{proof}

In the sequel we will need the following notions. The {\em home} of an agent is the last node at which it arrived applying 
procedure {\tt Asynch} and the {\em cottage} of an agent is the previous node that it reached applying 
procedure {\tt Asynch} (or the starting node of the agent, if the home is reached in the first step of procedure {\tt Asynch}).
The home of an agent $a$ is denoted by $Home(a)$ and its cottage by $Cot(a)$.
Hence if agent $a$ arrives at $y$ from $x$ applying a step of procedure {\tt Asynch}, then $x=Cot(a)$ and $y=Home(a)$.

We say that a fault occurring during the execution of Algorithm {\tt RV-RF} is {\em repaired} if either the execution of phase Correction
following this fault has been completed without any occurrence of a fault, or the execution of phase Correction
following this fault has been interrupted by a fault that has been repaired. Intuitively, this recursive definition says that a fault is repaired if it is followed by a series
of partial executions  of phase Correction, each except the last one interrupted by a fault, and the last one not interrupted by a fault and executed completely.

The following two lemmas show that after the occurrence of event 1 or event 3 from Lemma \ref{5 events} a meeting of the agents at some node must happen within  $O(\max(\log \lambda_1, \log \lambda_2))$ steps in phase Progress.

\begin{lemma}\label{ab}
Consider two agents $a$ and $b$, with different labels $\lambda_1$ and $\lambda_2$, respectively,  starting at arbitrary different nodes of an $n$-node graph, where $n$ is unknown to the agents. Suppose that in some round $t'$ of the execution of Algorithm {\tt RV-RF} 
in the model $\cal M$ we have $Home(a)=Cot(b)$ and $Home(b)=Cot(a)$. Then a meeting
occurs at some node after $O(\max(\log \lambda_1, \log \lambda_2))$ steps of the phase Progress executed after round $t'$.
\end{lemma}

\begin{proof}
Suppose, without loss of generality, that round $t'$ is the first round for which {$Home(a)=Cot(b)=y$ and $Home(b)=Cot(a)=x$}.
This implies that in some round $t \leq t'$ one of the agents, say $a$, came from $x$ to $y$ by applying a step of procedure {\tt Asynch}
and the following conditions are satisfied:
\begin{itemize}
\item
agent $a$ has not performed the next step of procedure {\tt Asynch} in the time interval $[t,t']$;
\item
agent $b$ came from $y$ to $x$ by applying a step of procedure {\tt Asynch} in round $t'$.
\end{itemize}
If agent $a$ completed the procedure {\tt Dance} following its last step of {\tt Asynch} before round $t'$,
then a meeting occurs in round $t'-1$, as both agents are at node $y$ in this round. Hence in the sequel we suppose that agent $a$
has not yet completed this procedure  {\tt Dance} in round $t'-1$. 

Let $t'' \geq t'$ be the first round in which one of the agents completes its procedure {\tt Dance}.
First suppose that $t'$ and $t''$ are rounds in which none of the agents is subject to a fault and in which none of the agents executes a step of phase Correction. This means that if a fault occurred in the time interval $(t,t')$, then it must have been repaired in this time interval.

\noindent
{\bf Claim 1.} If no fault occurs in the time interval $(t',t'')$, then a meeting occurs at some node by round $t''$.

In order to prove the claim, first suppose that $t=t'$. In this case the agents start executing procedure {\tt Dance} simultaneously.
By the definition of modified labels, there is an index $i$ such that bits at position $i$ of the modified labels of the agents differ.
Hence by round $t''$ one of the agents is inert at one of the nodes $x$ or $y$ while the other traverses the edge $\{x,y\}$. This
implies a meeting at $x$ or at $y$.

Next suppose that $t<t'$ and consider 3 cases.

Case 1. In round $t'$ agent $a$ has not yet finished Stage 1 of procedure {\tt Dance}.

In round $t'-1$ both agents were at node $y$, hence the meeting occurred.

Case 2. In round $t'$ agent $a$ has already finished Stage 1 of procedure {\tt Dance} but has not yet finished Stage 2.

Between round $t'+1$ and round $t'+10$ agent $b$ executes Stage 1 of procedure {\tt Dance} and hence is inert at $x$.

Subcase 2.1. Agent $a$ spends all the time interval $[t'+1,t'+10]$ in Stage 2 of procedure {\tt Dance}.

By Lemma \ref{10}, in at least one of these rounds agent $a$ traverses the edge $\{x,y\}$, while agent $b$ is inert at $x$ during all these
rounds. Hence the meeting must occur at $x$.

Subcase 2.2. Agent $a$ starts Stage 3 of procedure {\tt Dance} in the time interval $[t'+1,t'+10]$.

Hence during round $t'+10$ agent $a$ executes Stage 3 of procedure {\tt Dance} and hence traverses edge $\{x,y\}$,
while agent $b$ is inert at $x$ during this
round. Hence the meeting must occur at $x$.

Case 3. In round $t'$ agent $a$ has already finished Stage 2 of procedure {\tt Dance}.

Subcase 3.1. In round $t'$ agent $a$ has not yet finished Stage 3 of procedure {\tt Dance}.

If in round $t'$ agent $a$ is at  $x$, then the meeting occurs in this round. Otherwise,
agent $a$ moves to $x$ in round $t'+1 \leq t''$ and the meeting occurs at $x$ in this round.

Subcase 3.2. In round $t'$ agent $a$ finished Stage 3 of procedure {\tt Dance}.

Agent $b$ started  Stage 3 of its preceding procedure {\tt Dance} at node $y$ one round before agent $a$ started
Stage 3 at $y$. When agent $b$ starts treating the last bit 0 of its transformed label and hence waits at node $y$
in round $t'-13$, agent $a$ has not finished treating its penultimate bit 1 and hence traverses edge $\{x,y\}$ in this round.
If agent $a$ is at $x$ in round $t'-14$, this implies a meeting at $y$ in round $t'-13$. If agent $a$ is at $y$ in round $t'-14$,
this implies a meeting at $y$ in round $t'-14$, as $b$ was at $y$ in this round.

This concludes the proof of Claim 1.

\noindent
{\bf Claim 2.} If some faults occur in the time interval $(t',t'')$, then a meeting occurs at some node by round $t''$.

Consider two cases.

Case 1. If an agent is subject to a fault in some round in the time interval $(t',t'')$, then the other agent is also subject to a fault in this round.

In this case the rounds of executing the phase Correction are the same for both agents, hence
we can ``delete'' them and reduce the situation to the case when no faults occur in the time interval $(t',t'')$. The claim follows from Claim~1. 

Case 2. There exists a round in the time interval $(t',t'')$, in which exactly one agent is subject to a fault.

Let $r$ be the last round in the time interval $(t',t'')$, in which exactly one agent is subject to a fault. Denote by $B$ the time interval
$[r+1,r+20]$ and let $f$ be the agent subject to a fault in round $r$. {Note that $B$ is necessarily included in the time interval $(t',t'')$, since we first assume in the beginning of the proof of this lemma that if a fault occurred in the time interval $(t,t')$, then it must have been repaired in this time interval.} 

During the time interval $B$ the agent $f$ is inert because these
are the first 20 rounds of phase Correction, and hence it is not subject to any faults. $B$ is included in the time interval $(t',t'')$
because any fault occurring in the time interval  $(t',t'')$ must be repaired in this time interval. If the other agent $f'$ executes procedure 
 {\tt Dance} during all rounds of $B$, then a meeting must occur at some node, because this agent must move at least once in the time interval $B$
 by Lemma \ref{20}. Hence we may assume that there exist rounds in $B$ in which $f'$ does not execute procedure  {\tt Dance}.
 
 Agent $f'$ cannot  be subject to a fault in the time interval $B$ because it would be subject to such a fault alone, 
 contradicting the definition of round $r$. Moreover, agent $f'$ cannot be subject to a fault in round $r$ by the definition of this round.
 Hence agent $f'$ must execute during some rounds of $B$ a part of phase Correction
 caused by a fault occurred before round $r$. This implies that agent $f'$ executes some round of Stage 2 or Stage 3 of phase Correction during the time interval $B$. Since Stage 2 and Stage 3 contain only steps consisting in traversing edge $\{x, y\}$, the agents must meet
at some node by the end of time interval B.
 
 This concludes the proof of Claim 2.
 
 The two claims imply that if $t'$ and $t''$ are rounds in which none of the agents is subject to a fault and in which none of the agents executes a step of phase Correction, then a meeting must occur at some node by round $t''$. It remains to consider the case when this condition is not satisfied. Note that in rounds $t'$ and $t''$ at least one of the agents moves. Hence if a fault occurs in one of these rounds, this means
 that one of the agents traverses edge  $\{x,y\}$ while the other agent is idle at $x$ or at $y$. This implies a meeting at some node by round $t''$.
 
 Hence we may assume that an agent executes a step of phase Correction either in round $t'$ or in round $t''$.  
 
 Case 1. An agent executes a step of phase Correction in round $t'$.
 
 In this case agent $a$ executes a  step of phase Correction in round $t'$ while agent $b$ comes to $x$ from $y$ executing a step of procedure {\tt Asynch}. Suppose, for contradiction, that the agents do not meet at some node by round $t''$. In round $t'$ agent $a$ must move, otherwise a meeting occurs at some node either in round $t'-1$ or in round $t'$. Moreover, in round $t'+1$ agent $a$ must attempt to move. Indeed, either
 agent $a$ has not finished the phase Correction in round $t'$, in which case it continues a moving attempt in round $t'+1$, or 
 it finished the phase Correction in round $t'$, in which case in round $t'+1$ it resumes phase Progress where it was interrupted
 by the last fault, i.e., it also attempts to move.
 
 If agent $a$ is not subject to a fault in round $t'+1$, then it traverses edge  $\{x,y\}$ in this round. In round $t'+1$ agent $b$ executes the
 first round of procedure {\tt Dance}, hence it is idle at $x$. This implies that a meeting occurs at some node either in round $t'$ or in round $t'+1$.
 Hence we may assume that agent $a$ is subject to a fault in round $t'+1$.  In the time interval $[t'+2,t'+21]$ agent $a$ executes the
 first 20 rounds of phase Correction following this fault, i.e., it remains idle. On the other hand, agent $b$ completes 
 Stage 1 of procedure {\tt Dance} (which is a waiting period) in round $t'+10$ and makes an attempt to move in round $t'+11$.
 If it is not subject to a fault in this round, a meeting at some node must occur. Hence we may assume that agent $b$ is subject to a fault in round $t'+11$.
 This implies that in the time interval $[t'+12,t'+31]$ agent $b$ executes the
 first 20 rounds of phase Correction following this fault, i.e., it remains idle. On the other hand agent $a$ completes Stage 1 of phase Correction in round $t'+21$ and makes an attempt to move in round $t'+22$. If it is not subject to a fault in this round, a meeting
 at some node must occur,
 because $b$ is idle in this round. Hence we may assume that $a$ is subject to a fault in round $t'+22$ and executes the
 first 20 rounds of phase Correction following this fault in the time interval $[t'+23,t'+42]$, i.e., it remains idle in this time interval.
 Continuing this reasoning we conclude that none of the agents can finish its procedure {\tt Dance} by round $t''$, which is a contradiction.
 
  Case 2. An agent executes a step of phase Correction in round $t''$.
  
  Let $f$ be the agent finishing its procedure {\tt Dance} in round $t''$ and let $f'$ be the agent executing a step of phase Correction in round $t''$. Suppose, for contradiction, that the agents do not meet at some node by round $t''$. Agent $f$ moves in each round of the time interval $[t''-11, t'']$.
  Indeed, agent $f$ could not be subject to a fault in one of these rounds or execute Stage 1 of the phase Correction, because
  then it could not finish procedure {\tt Dance} in round $t''$. 
 For the same reasons, agent $f$ could not execute any steps of Stage 1 or Stage 2 of procedure {\tt Dance}. 
  Hence in these rounds agent $f$ either executes the entire Stage 3
  of procedure {\tt Dance}, or executes some rounds of Stage 2 or 3 of phase Correction and then at least one round of procedure {\tt Dance}.
  Hence in all these rounds the agent must move. It follows that in each round of the time interval $[t''-11, t'']$ agent $f'$ also moves,
  for otherwise there would be a meeting at some node by round $t''$. Hence agent $f'$ executes rounds of Stage 2 or Stage 3 of phase Correction
  in this time interval. It follows that in round $t''-21$ agent $f'$ is idle executing a step of Stage 1 of phase Correction.
  There are two subcases.
  
  Subcase 2.1. Agent $f$ does not terminate any phase Correction in the time interval $[t''-15,t'']$.
  
  During this time interval agent $f$ executes exclusively steps of procedure {\tt Dance}, which it finishes in round $t''$. Hence in round
  $t''-12$ agent $f$ executes the last step of Stage 2 of this procedure, in which it is idle. Agent $f'$ must be also idle during this round,
  otherwise a meeting occurs at some node. Since in round $t''-11$ agent $f'$ moves executing a step of the phase Correction, this implies that in round 
  $t''-12$ agent  $f'$ executes the last step of Stage 1 of phase Correction. Hence $f'$ is idle in round $t''-15$. In this round agent $f$ treats 
  the penultimate bit of its modified label in Stage 2 of procedure {\tt Dance}. This bit is 1. Hence agent $f$ traverses edge $\{x,y\}$
  in round $t''-15$. This implies a meeting at some node in round $t''-16$ or $t''-15$, which is a contradiction.
  
  Subcase 2.2. Agent $f$ terminates a phase Correction in the time interval $[t''-15,t'']$.
  
  If this happens in round $t''-2$ or earlier, then in round $t''-21$ agent $f$ executes a step of Stage 2 or Stage 3 of phase Correction.
  Hence in this round agent $f$ traverses edge $\{x,y\}$, while agent $f'$ is idle. 
  This implies a meeting at some node in round $t''-22$ or $t''-21$, which is a contradiction. Since agent $f$ cannot end phase Correction in round $t''$ by definition of this round, we may assume that it terminates phase Correction in round $t''-1$. If $f$ executed 21 rounds of movement
  in  Stage 2 and Stage 3 of phase Correction, then a meeting at some node must occur by round $t''-21$. Otherwise agent $f$ executed exactly 20 moves
  in phase Correction and hence was subject to a fault in round $t''-41$. Consider two possibilities.
  If  both  agents are always subject to faults simultaneously in the time interval $[t',t'']$, then
  both agents execute steps of phase Correction in the same rounds in the time interval $[t',t'']$, which
  contradicts the fact that in round $t''$ agent $f$ executes the last step of procedure {\tt Dance} and agent $f'$ executes a step
  of phase Correction in this round. If agents are not always subject to faults simultaneously in the time interval $[t',t'']$, then
  the contradiction is obtained using an argument similar to that from Case 2 in the proof of Claim 2.
  
  
  
  Hence we proved that a meeting at some node must occur by round $t''$. Since in the time interval $[t',t'']$ both agents
  executed $O(\max(\log \lambda_1, \log \lambda_2))$ steps of the phase Progress, the proof is complete. 
  
  \end{proof}
  
  The proof of the next lemma is analogous to that of Lemma \ref{ab}, hence we omit it.

\begin{lemma}\label{aa}
Consider two agents $a$ and $b$, with different labels $\lambda_1$ and $\lambda_2$, respectively,  starting at arbitrary different nodes of an $n$-node graph, where $n$ is unknown to the agents. Suppose that in some round $t'$ of the execution of Algorithm {\tt RV-RF} 
in the model $\cal M$ we have {$Home(a)=Home(b)$ and $Cot(a)=Cot(b)$}. Then a meeting
occurs at some node after $O(\max(\log \lambda_1, \log \lambda_2))$ steps of the phase Progress executed after round $t'$.
\end{lemma}

The last lemma of this section shows that the meeting between agents in model $\cal M$ is guaranteed after a polynomial number of steps in phase Progress.

\begin{lemma}\label{final}
Consider two agents $a$ and $b$, with different labels $\lambda_1$ and $\lambda_2$, respectively,  starting at arbitrary different nodes of an $n$-node graph, where $n$ is unknown to the agents. There exists a polynomial $B$ such that a meeting at some node is guaranteed after the execution of a total of $B(n,\max(\log \lambda_1, \log \lambda_2))$ steps of phase Progress of  Algorithm {\tt RV-RF} 
by both agents in the model $\cal M$.
\end{lemma}

\begin{proof}
Each stage of phase Progress of  Algorithm {\tt RV-RF}  consists of one step of  procedure {\tt Asynch}  and one execution of procedure {\tt Dance}, and the
execution time of procedure {\tt Dance} is logarithmic in the label of the executing agent. Hence after executing a total of $A(n,\min(\log \lambda_1, \log \lambda_2))$ stages of phase Progress the agents with labels $\lambda_1$ and $\lambda_2$ executed a total of
at most 
$$B^*(n,\max(\log \lambda_1, \log \lambda_2))=A(n,\min(\log \lambda_1, \log \lambda_2)) \cdot c \cdot \max(\log \lambda_1, \log \lambda_2)$$ 
steps of phase Progress, for some positive constant $c$. By Lemma \ref{5 events}, after executing  a total of $A(n,\min(\log \lambda_1, \log \lambda_2))$ stages of phase Progress
by both agents, one of the 5 events must have occurred.

If event 1 occurred, the conclusion follows from Lemma \ref{ab} for
\begin{equation}
{B(n,\max(\log \lambda_1, \log \lambda_2))=B^*(n,\max(\log \lambda_1, \log \lambda_2))+O(\max(\log \lambda_1, \log \lambda_2))}
\label{equa}
\end{equation} 
If event 2 occurred, the conclusion is immediate. If event 3 occurred,
the conclusion follows from Lemma \ref{aa}, for $B$ given by {\eqref{equa}}. If event 4 occurred, there are two cases.

Case 1. Both agents leave node $v$ in the same round applying a step of procedure {\tt Asynch}.

In this case the agents were together at $v$ in the previous round, hence rendezvous occurred and the conclusion follows from the fact that each agent has executed at most $c\cdot \max(\log \lambda_1 , \log \lambda_2)$ steps of procedure {\tt Dance} since it arrived at $v$. 

Case 2. One of the agents leaves node $v$ applying a step of procedure {\tt Asynch} in a round $m$ before the other agent
leaves node $v$ applying a step of procedure {\tt Asynch}.

In this case in round $m+1$ we have $Home(a)=Cot(b)$ and $Home(b)=Cot(a)$ and hence the conclusion follows from Lemma \ref{ab},
for $B$ given by {\eqref{equa}}.

Finally, if event 5 occurred, there are two cases.

Case 1. Both agents leave node $v$ in the same round applying a step of procedure {\tt Asynch}.

In this case the agents were together at $v$ in the previous round, hence rendezvous occurred and the conclusion follows from the fact that each agent has executed at most $c\cdot \max(\log \lambda_1 , \log \lambda_2)$ steps of procedure {\tt Dance} since it arrived at $v$.

Case 2. Agent $a$ leaves node $v$ applying a step of procedure {\tt Asynch} in a round $m$ before agent $b$
leaves node $v$ applying a step of procedure {\tt Asynch}.

In this case in round $m+1$ we have $Home(a)=Cot(b)$ and $Home(b)=Cot(a)$ and hence the conclusion follows from Lemma \ref{ab},
for $B$ given by  {\eqref{equa}}.

\end{proof}

We are now ready to prove the main result of this section, showing that Algorithm {\tt RV-RF} achieves rendezvous at polynomial cost 
with very high probability, under the random fault model.

\begin{theorem}
Consider two agents $a$ and $b$, with different labels $\lambda_1$ and $\lambda_2$, respectively,  starting at arbitrary different nodes of an $n$-node graph, where $n$ is unknown to the agents. Suppose that delay faults occur randomly and independently with constant probability
$0<p<1$ in each round and for each agent.
Algorithm {\tt RV-RF} guarantees rendezvous of the agents with probability 1. Moreover, there exists a polynomial $B$ such that rendezvous at some node occurs at cost $\tau=O(B(n,\max(\log \lambda_1, \log \lambda_2)))$ with probability at least $1-e^{-O(\tau)}$.
\end{theorem}

\begin{proof}
An entire execution of phase Correction of Algorithm {\tt RV-RF}, if no fault occurs during it, lasts at most 41 rounds.
Hence in a segment of 42 consecutive rounds without a fault an agent makes at least one step of phase Progress.
Call a segment of 42 consecutive rounds without a fault occurring to a given agent {\em clean} for this agent. For any round
$r$ the probability that the segment starting at $r$ is clean for both agents is at least $q=(1-p)^{84}$.
Let $B$ be the polynomial from Lemma \ref{final}. The probability of the
existence of $B(n,\max(\log \lambda_1, \log \lambda_2))$ pairwise disjoint clean segments is 1, and this event implies rendezvous
by Lemma \ref{final}.

Let $c=\lceil 84/q \rceil$. Consider $\tau=c\cdot B(n,\max(\log \lambda_1, \log \lambda_2))$ rounds of execution of Algorithm {\tt RV-RF}.
Partition these rounds into $\lceil 2/q \rceil B(n,\max(\log \lambda_1, \log \lambda_2))$ pairwise disjoint segments of length 42.
 By Chernoff bound, there are 
at least $B(n,\max(\log \lambda_1, \log \lambda_2))$ clean segments
among them, with probability $1-e^{-O(\tau)}$.
This concludes the proof.
\end{proof}

We close this section by sketching an alternative solution to the problem of rendezvous at polynomial cost, with very high-probability, in the presence of random delay faults.
This alternative solution relies on the well-known fact that two random walks operating in an $n$-node network, without faults, meet after polynomial time with probability $1-n^{-c}$, for some positive constant $c$. This follows, e.g., from the fact that the cover time of a random walk is polynomial in any graph and from \cite{aldous}. By suitably changing the polynomial, the probability can be raised to very high, i.e., the {error probability decreased} to inverse-exponential
in $n$. The problem with applying this fact is that it concerns random walks, i.e., an algorithm in which {the agents} have access to random bits, while we want a deterministic algorithm working
with very high probability in the presence of random independent delay faults. Thus the issue is how to simulate a random walk by a deterministic algorithm in this faulty environment. Not surprisingly, this
can be done by harvesting  randomness from the random independent faults occurring in the network.

We will use a procedure called {\em unbiased random bit production} ({\tt URBP}), which is executed by an agent.  Let $u$ be a position of the
agent and let $v$ be an adjacent node. The procedure works in two rounds as follows.  In each round the agent attempts a move along the edge $\{u,v\}$. 
If exactly one of these two attempts is successful, the procedure outputs bit 1 if the first attempt is successful and outputs bit 0 if the second attempt is successful.
Otherwise, the procedure outputs nothing.
 Clearly, a bit obtained in this way is unbiased (has probability 1/2).

Fix an integer $n$ and let $Q(n)$ be a polynomial such that two random walks meet in any $n$-node network, with very high probability, after time at most $Q(n)$. The deterministic algorithm based on random walks works in epochs numbered
by consecutive positive integers $n$. Each epoch is divided into two parts: {\em bit preparation} and {\em execution}. {The bit preparation part of the $n$th epoch consists of repeating
procedure {\tt URBP} $nk \cdot Q(n)$ times, where $k=\lceil \log n \rceil$. By Chernoff bound, for a sufficiently large constant $q$ depending on fault probability $p$,
repeating procedure {\tt URBP} $qk \cdot Q(n)$ times is enough in order to produce $k \cdot Q(n)$ random independent unbiased bits with very high probability. Hence, for $n$ large enough,  at least $k \cdot Q(n)$
such bits are obtained with very high probability at the end of the bit preparation part of the $n$th epoch.}

The execution part of the $n$th epoch {is} as follows. In a given round of this epoch some prefix of the sequence of $k \cdot Q(n)$ bits prepared in the bit preparation part is already {\em used}. 
The agent, currently located at a node of degree $d$ takes the first unused $r$ bits, where $r=\lceil \log d \rceil$. If this string of bits is the binary representation of an integer
 $i<d$, then the agent attempts a move by port $i$. Otherwise the agent stays at the current node. The string of $r$ bits is added to the prefix of used bits. The execution
 part of the $n$th epoch lasts $Q(n)$ rounds. If the total string of prepared bits is of length at least $k \cdot Q(n)$, there are enough bits to execute $Q(n)$ rounds of the execution phase, regardless of the degrees of the nodes.
If there are not enough bits to perform the execution part, the agent waits inert at the node at which it ran out of the bits to perform the rest of the random walk, until $Q(n)$ rounds of the execution part have elapsed.  
If rendezvous does not occur by the end of the $n$th epoch, the agent starts the $(n+1)$th epoch.
 
 It follows from the above description that the execution part of the $n$th epoch is a random walk in which the agent, currently located at a node of degree $d$
 takes each port with equal probability $(1-p)/2^{\lceil \log d \rceil}$ and stays inert with probability $1-d(1-p)/2^{\lceil \log d \rceil}$.
 Let $n$ be the size of the graph (unknown to the agents). If the agent woken up earlier
  executes its $n$th epoch before the later agent is woken up, it will meet this agent at its starting node with very high probability, at polynomial cost.
{Otherwise, when the later agent wakes up, the earlier agent executes a round of the $k$-th epoch for some $k\leq n$. Hence, there exists a polynomial $K(n)$ for which by the time the earlier agent executes the last round of the $K(n)$th epoch, both agents performed simultaneously at least $Q(n)$ consecutive rounds of the random walk during execution parts of some epochs, with very high probability, leading to a meeting with very high probability.} The total cost incurred by both agents until then is polynomial in $n$.
  
  We would like to argue that, while both methods solve the problem of rendezvous with random delay faults, our main method, i.e., Algorithm {\tt RV-RF}, has the advantage of holding
  under more general assumptions concerning faults.
  Indeed, it is based on the fact that, with random delay faults,
  each agent can execute
  with very high probability,  sufficiently many times a block of at least $42$ consecutive rounds without being subject to any fault. This  permits the agents to execute sufficiently many steps of phase Progress to achieve
rendezvous because an entire execution of phase Correction consists of at most $41$ rounds. In fact, Algorithm {\tt RV-RF} also works correctly at polynomial cost
in any network in which delay faults are not random but occur rarely, say never more frequently than once every 50 rounds. In this case Algorithm {\tt RV-RF}
 {\em always}  guarantees a meeting at polynomial cost, while the method based on random walks cannot be used.
 
 {Another advantage of Algorithm {\tt RV-RF} over  the algorithm based on harvesting randomness is the following.\footnote{{We are grateful to the anonymous referee for pointing this out.}}
 Algorithm {\tt RV-RF} has the intuitively desirable property that it works faster as the fault probability $p$ decreases, while harvesting randomness for very small $p$ is very slow.}

\section{Unbounded adversarial faults}

In this section we consider the scenario when the adversary can delay each of the agents for any finite number of consecutive rounds.
Under this scenario the time (number of rounds until rendezvous) depends entirely on the adversary, so the only meaningful measure
of efficiency of a rendezvous algorithm is its cost. However, it turns out that, under this harsh fault scenario, even feasibility of rendezvous is usually
not guaranteed, even for quite simple graphs. Recall that we {\em do not} assume knowledge of any upper bound on the size of the graph. 

In order to prove this impossibility result, we introduce the following terminology.
For any rendezvous algorithm $\cA$ in a graph $G$, a {\em solo} execution of this algorithm by an agent $A$ 
is an execution in which $A$ is alone in the graph. Note that, in an execution of algorithm $\cA$ where two
agents are present, the part of this execution by an agent before the meeting is the same as the respective part of a solo execution of the algorithm by the respective agent. 

A port labeling in a graph is called {\em homogeneous}, if port numbers at both endpoints of each edge are equal.

{\begin{remark}
\label{rem:remark1}
{
Let $\cF$ be a family of regular graphs of degree $d$, with homogeneous port labeling. Consider a solo execution of a rendezvous algorithm $\cA$ by an agent $A_{\lambda}$ with label $\lambda$, in a graph $G\in \cF$. In each round $r$ the agent may either decide to stay at the current node or try to traverse an edge by choosing some port. This decision depends on the history of the agent before round $r$, i.e., on the entire
knowledge it has in this round.
In view of the regularity of the graph, and since the port labeling is homogeneous, the agent does not learn anything about the graph during navigation: it can differentiate neither $G$ from any graph $G'\in \cF$, nor any two nodes in $G$. It follows that the history of the agent before round $r$ can be coded by the sequence of previous round numbers in which the agent made a move.
(The agent made a move in a previous round, if and only if, it tried to traverse an edge and the adversary allowed the move by not imposing a fault in this round.)
Hence, if the algorithm, the agent's label and the adversary's behavior are fixed,  the agent's history, in any round,
is necessarily the same in any graph of $\cF$ and for any starting node.  }
\end{remark}}

Consider the decisions of agent  $A_{\lambda}$
starting in some round $t$ and before its next move.
We say that the agent {\em attacks} since round $t$, {if, from round $t$ on,}
it tries an edge traversal in each round until it makes the next move. {Note that, within a given attack, the edge which the attacking agent attempts to traverse, is not always necessarily the same.}
{We say that the agent does not attack after round $t$, if there does not exist a round $t'>t$ in which the agent starts attacking. In other words} it means that, {after round $t$, the agent} 
decides to stay at the current node in rounds with arbitrarily large numbers, interleaved with possible attempts at a move,
all of which can be prevented by an adversary.
Note that an adversary must eventually allow an agent that attacks to make a next move but it
is capable of preventing {an agent that never attacks after a given round} from making any further move by imposing faults in all rounds in which the agent tries an edge traversal.


We first prove the following technical lemma.

\begin{lemma}\label{lem-adversaire}
Let $\cF$ be a family of regular graphs of degree $d$, with homogeneous port labeling, $\cA$ any rendezvous algorithm
working for graphs in  $\cF$, and $M$ any positive integer.
Then there exists {a behavior of the adversary inducing}
an infinite increasing sequence $\tau_1, \tau_2, \ldots, \tau_i, \ldots$ of positive integers, such that, for any graph $G \in \cF$ and 
for any $\lambda \leq M$, there exists a solo execution of algorithm $\cA$ in $G$  by agent $A_{\lambda}$ with label $\lambda$,
such that one of the following conditions is satisfied:

\begin{itemize}
\item {either  $A_{\lambda}$ attacks only a finite number $k\geq 0$ of times,
in which case $A_{\lambda}$ either does not move if $k=0$, or moves exactly in rounds 
$\tau_1,\tau_2, \ldots,
\tau_k$ otherwise;}
\item or $A_{\lambda}$ attacks infinitely many times, in which case it moves exactly in rounds $\tau_i$, for all positive integers $i$.
\end{itemize}
\end{lemma}

\begin{proof}
{Fix a rendezvous algorithm $\cA$, a  degree $d$ of graphs in a family $\cF$, and a positive integer $M$.
Consider the set $S$ of agents $A_{\lambda}$ with label $\lambda \leq M$, and suppose that each of these agents executes $\cA$ alone in some graph $G \in \cF$.
Note that, for a given behavior of an adversary and a given label, this execution will be the same in any graph $G \in \cF$ (cf. Remark~\ref{rem:remark1}).}
We construct the sequence $\tau_1, \tau_2, \ldots, \tau_i, \ldots$ and the behavior of an adversary by induction on $i$.
We start by giving a partial description of this behavior. For any agent in $S$, the adversary imposes a fault in every round in which
the agent does not attack but tries to traverse an edge. An adversary behaving in this way will be called {\em tough}.

If an agent does not attack before its first move, then the agent never moves, in view of the toughness of the adversary. 
Consider the set $S_0 \subseteq S$ of agents that attack before their first move.
Let $s_1<\cdots <s_r$ be the rounds in which this attack starts for some agent in $S_0$. Define
$\tau_1=s_r$. The adversary allows a move of all agents from the set $S_0$ in round $\tau_1$. 


For the inductive step,
consider the set  $S_{i-1}\subseteq S$ of agents that attacked $i>0$ times and made $i$ moves, exactly in rounds $\tau_1, \tau_2, \ldots, \tau_i$.
If an agent from $S_{i-1}$  does not attack after its $i$-th move in round $\tau_i$, then the agent never moves again, in view of the toughness of the adversary.
Consider the set $S_i \subseteq S_{i-1}$ of agents that attack after round $\tau_i$ and before the next move.
Let $t_1<\cdots <t_p$, where $t_1>\tau_i$, be the rounds in which the $(i+1)$-th attack starts for some agent in $S_i$. Define
$\tau_{i+1}=t_p$. The adversary allows a move of all agents from the set $S_i$ in round $\tau_{i+1}$.

This concludes the definition of the sequence $\tau_1, \tau_2, \ldots, \tau_i, \ldots$, and of the behavior of the adversary, by induction on $i$. 
In order to prove that one of the conditions in the lemma must be satisfied, consider the solo execution $\cE_{\lambda}$ of an agent  $A_{\lambda}$ with label $\lambda \leq M$, in some graph $G \in \cF$,
against the above defined adversary. There are two possible cases: either, in the execution $\cE_{\lambda}$, the agent attacks $k$ times, where $k$ is a non-negative integer, or it attacks infinitely many times. By the definition of the adversary, in the first case the agent {does not move or} moves exactly in rounds 
$\tau_1,\tau_2, \ldots,\tau_k$, and in the second case it moves exactly in rounds $\tau_i$, for all positive integers $i$.

\end{proof}

We now establish the impossibility result for the model with unbounded faults.

\begin{theorem}\label{impossible}
Rendezvous with unbounded adversarial faults is not feasible, even in the class of rings.
\end{theorem}

\begin{proof}

{Let $\cF$ be the family of rings of even size with homogeneous port numbering. Since graphs in $\cF$ are regular graphs of degree $2$, Lemma~\ref{lem-adversaire} applies.
Suppose that $\cA$ is a rendezvous algorithm working for the family $\cF$, and let $M=3$.
Let $X=(\tau_1, \tau_2, \ldots, \tau_i, \ldots)$ be the sequence of integers given by Lemma \ref{lem-adversaire}, and
consider the solo execution $\cE_j$ of agent $A_j$ with label $j$, for $j=1,2,3$, corresponding to $X$ and satisfying one of the two conditions of Lemma~\ref{lem-adversaire}. (For a fixed $j\leq 3$, this execution 
is the same in each ring from $\cF$, regardless of the starting node.)
Consider two cases.}

{Case 1. At least two among agents $A_j$ attack infinitely many times in execution $\cE_j$.}

{Without loss of generality, assume that $A_1$  and $A_2$ attack infinitely many times, each in their solo execution $\cE_1$ (resp. $\cE_2$).
Consider the execution $\cE$ in which agents $A_1$ and $A_2$ start simultaneously at an odd distance in the graph, and the adversary acts against each of them as in their respective solo execution $\cE_1$ and $\cE_2$. In execution
$\cE$, both agents make moves in exactly the same rounds, i.e., the rounds of sequence $X$. Since an even ring is bipartite, the parity of their distance never changes, and they remain at an odd distance forever. Hence they never meet, which is a contradiction.}

{Case 2. At least two among agents $A_j$ attack finitely many times in execution $\cE_j$.}

{Without loss of generality, assume that $A_1$ and $A_2$ attack finitely many times, each in their solo execution $\cE_1$ (resp. $\cE_2$): agent $A_1$ attacks $k$ times and
agent $A_2$ attacks $k'$ times. Choose as the starting nodes of these agents antipodal nodes in the ring
of size $2(k+k'+1)$ (i.e., at distance $k+k'+1$), start them simultaneously, and assume that the adversary acts against each of them as in their respective solo execution $\cE_1$ and $\cE_2$. By Lemma \ref{lem-adversaire}, agent $A_1$ will make a total number of $k$ moves, and agent $A_2$
will make a total number of $k'$ moves, in this execution. Hence they never meet, which is a contradiction.}

\end{proof}

In view of Theorem \ref{impossible}, it is natural to ask if rendezvous with unbounded adversarial faults can be accomplished
in the class of connected graphs not containing cycles, i.e., in the class of trees, and if so, at what cost it can be done.
The rest of this section is devoted to a partial answer to this problem. We start with an auxiliary result about rendezvous
in oriented rings, under the {\em additional} assumption that the size of the ring is known to the agents. This result will be used
in a special situation when the size of the ring can be {\em learned} by the agents at small cost. By an oriented ring we mean a ring
in which every edge has ports 0 and 1 at its extremities. This provides the agents with the orientation of the ring: they can both
go in the same direction by choosing port 0 at each node.

\begin{lemma}\label{oriented}
If the size $n$ of an oriented ring is known to the agents, then rendezvous with unbounded adversarial faults can be achieved at cost
$O(n\ell)$, where $\ell$ is the smaller label.
\end{lemma}

\begin{proof}
The rendezvous algorithm for an agent with label $\lambda$ is the following: start by port 0, perform $2n\lambda$ edge traversals in the same direction, and stop. If two agents execute this algorithm, then, if they have not met before, the agent with larger label does at least
one full tour of the ring after the agent with smaller label already stopped. Hence they must meet at cost at most $4(\ell+1)n$,
where $\ell$ is the smaller label. 

\end{proof}

Our goal is to present an efficient rendezvous algorithm working for arbitrary trees. We will use the following notion.
Consider any tree $T$.
A {\em basic walk} in $T$, starting from node $v$ is a traversal of all edges of the tree ending at the starting node $v$ and defined as follows.
Node $v$ is left by port 0; when the walk enters a node by port $i$, it leaves it by port $(i+1)$ mod $d$, where $d$ is the degree of the node.
Any basic walk consists of $2(n-1)$ edge traversals.
An agent completing the basic walk knows that this happened and learns the size $n$ of the tree and the length $2(n-1)$ of the basic walk.

\remove
{
Consider the following sequence of trees constructed recursively from $T$: $T_0=T$, and $T_{i+1}$ is the tree obtained from
$T_i$ by removing all its leaves. $T'=T_j$  for the smallest $j$ for which $T_j$ has at most two nodes. If $T'$ has one node, then
this node is called the {\em central node} of $T$. If $T'$ has two nodes, then the edge joining them is called the {\em central edge} of $T$.

After performing a basic walk starting from any node of a tree $T$ the agent has a complete map of the tree with all port numbers marked.
In particular, it knows whether the tree has a central node or a central edge. Define the {\em canonical walk} to be the basic walk starting at the central node in the first case
and, in the second case, as the walk starting at any extremity of the central edge, choosing the port corresponding to this edge and then following the rule that
when the walk enters a node by port $i$, it leaves it by port $(i+1)$ mod $d$, where $d$ is the degree of the node.
}

The following Algorithm {\tt Tree-RV-UF} (for rendezvous in trees with unbounded faults) works for an agent with label $\lambda$, starting at an arbitrary node of any tree~ $T$.

\smallskip

\noindent
{\bf Algorithm Tree-RV-UF}

\noindent
Repeat $2\lambda$ basic walks starting from the initial position and stop.
\hfill $\diamond$

\begin{theorem}\label{tree}
Algorithm {\tt Tree-RV-UF} is a correct rendezvous algorithm with unbounded adversarial faults in arbitrary trees, and works at cost $O(n\ell)$, where
$n$ is the size of the tree and $\ell$ is the smaller label.
\end{theorem}

\begin{proof}
Fix any nodes $v_1$ and $v_2$ of $T$, which are the starting positions of the agents.  Notice that agents repetitively performing a basic walk starting at $v_1$ and $v_2$ respectively, traverse all edges
of the tree in the same order and in the same direction, with a cyclic shift. Hence performing a repetitive basic walk in a tree of size $n$ can be considered as making tours of a ``virtual'' oriented ring of length $2(n-1)$ composed of edges in the order and direction imposed by the basic walk. Each edge of the tree is traversed exactly twice in each tour of this
virtual ring (once in each direction) and the direction of the walk in this virtual ring is the same, regardless of the node where the basic walk starts. Since, after completing the first basic walk, the agent learns
its length, and 
Algorithm {\tt Tree-RV-UF} is equivalent to the algorithm from the proof of Lemma \ref{oriented}  providing rendezvous in oriented rings of known size
(run on the virtual oriented ring given by the basic walk of the tree), the conclusion follows.

\end{proof}

We do not know if Algorithm {\tt Tree-RV-UF} has optimal cost, i.e., if a lower bound $\Omega(n\ell)$ can be proved on the cost of any
rendezvous algorithm with unbounded adversarial faults, working in arbitrary trees of size $n$. However, we prove a weaker lower bound.
It is clear that no algorithm can beat cost $\Theta(n)$ for rendezvous in $n$-node trees, even without faults. Our next result shows that,
for unbounded adversarial faults, $\Omega(\ell)$ is a lower bound on the cost of any rendezvous algorithm, even for the simplest tree,
that of two nodes.  

\begin{proposition}\label{lb}
Let $T$ be the two-node tree.
Every rendezvous algorithm with unbounded adversarial faults, working for the tree $T$, has cost $\Omega(\ell)$, where
$\ell$ is the smaller label. 
\end{proposition}

\begin{proof}
Let $M$ be a positive integer, and $\cA$ a rendezvous algorithm working in the tree $T$. For every $j\in \{1,\dots,M\}$, 
consider the solo execution $\cE_j$ of $\cA$, for agent $A_j$ with label $j$, and
for the adversary constructed in Lemma \ref{lem-adversaire}.
Let $\cE(i,j)$ be the execution of $\cA$ in $T$, for agents $A_i$ and $A_j$ starting simultaneously from the two nodes of $T$,
for the same adversary.
By  Lemma \ref{lem-adversaire}, for at most one value of $j\in \{1,\dots,M\}$, agent $A_j$ attacks infinitely many times
in execution $\cE_j$, because if there were two such agents $A_i$ and $A_j$, then they would not meet in execution  $\cE(i,j)$.
Hence, for at least $M-1$ values of $j\in \{1,\dots,M\}$, agent $A_j$ attacks a finite number of times in execution $\cE_j$.
Again by Lemma \ref{lem-adversaire}, for all these values of $j$, agent $A_j$ must attack a different number of times in execution
$\cE_j$, because if there were two agents $A_i$ and $A_j$ attacking the same number $k$ of times, then they would not meet in execution  $\cE(i,j)$. Hence, for at least two different labels $\ell <L \leq M$ this number of attacks must be at least $M-2$.
It follows from Lemma \ref{lem-adversaire} that in execution $\cE(\ell,L)$ each of the agents makes at least $M-2$ moves before
rendezvous. Since $\ell \leq M-1$, 
agent $A_{\ell}$ makes at least $\ell-1$ moves before rendezvous, which completes the proof.

\end{proof}

\section{Bounded adversarial faults}

In this section we consider the scenario when the adversary can delay each of the agents for at most $c$ consecutive rounds, where $c$ is a positive integer, called the {\em fault bound}.
First note that if $c$ is known to the agents, then, given any synchronous rendezvous algorithm working without faults for arbitrary networks, 
it is possible to obtain an algorithm working for bounded adversarial faults and for arbitrary networks, at the same cost.
Let $\cA$ be a synchronous rendezvous algorithm for the scenario without faults, working for arbitrary networks. Consider the following algorithm $\cA(c)$ working for
bounded adversarial faults with parameter $c$. Each agent replaces each round $r$ of algorithm  $\cA$ by a segment of $2c+1$ rounds.
If in round $r$ of algorithm $\cA$  the agent was idle, this round is replaced by $2c+1$ consecutive rounds in which the agent is idle. If in round $r$ the agent left
the current node  by port $p$, this round is replaced by a segment of $2c+1$ rounds in each of which the agent makes an attempt 
to leave the current node $v$
by port $p$ until it succeeds, and in the remaining rounds of the segment it stays idle at the node adjacent to $v$ that it has just entered. 

We associate the first segment of the later starting agent with the (unique) segment of the earlier agent that it intersects in at least $c+1$ rounds. Let it be the $i$th segment of the earlier agent.
We then associate the $j$th segment of the later agent with the $(j+i-1)$th segment of the earlier agent, for $j>1$. Hence, 
regardless of the delay between starting rounds of the agents, corresponding segments intersect in at least $c+1$ rounds. If the agents met at node $x$ in the $j$th round of the 
later agent, according to algorithm $\cA$, then, according to algorithm $\cA(c)$,  in the last $c+1$ rounds of its $j$th segment the later agent is at $x$ and
in the last $c+1$ rounds of its $(j+i-1)$th segment the earlier agent is at $x$. Since these segments intersect in at least $c+1$ rounds, there is a round in which both agents are at node $x$
according to algorithm $\cA(c)$, regardless of the actions of the adversary, permitted by the bounded adversarial fault scenario. This shows that algorithm $\cA(c)$ is correct.
Notice that the cost of algorithm $\cA(c)$ is the same as that of algorithm $\cA$, because in each segment corresponding to an idle round of algorithm $\cA$, an agent stays idle
in algorithm $\cA(c)$ and in each segment corresponding to a round in which an agent traverses an edge in algorithm $\cA$, the agent makes exactly one traversal in algorithm $\cA(c)$.

In the rest of this section we concentrate on the more difficult situation when the fault bound $c$ is unknown to the agents. The following 
Algorithm {\tt Graph-RV-BF} (for rendezvous in graphs with bounded faults) works for an agent with label $\lambda$ starting at an arbitrary node of any graph.

Algorithm {\tt Graph-RV-BF} is divided into phases. The $i$-th phase is composed of $2^i$ stages, each lasting
$s_i=2^{i+4}$ rounds. Hence the $i$-th phase lasts $p_i=2^{2i+4}$ rounds. The $\lambda$-th stage of the $i$-th phase consists of two parts: the \emph{busy} part of $b_i=3\cdot 2^{i}$ rounds and the \emph{waiting} part of  $w_i=13\cdot 2^{i}$ rounds. During the busy part of the $\lambda$-th stage of phase $i$, the agent tries to explore the graph three times ({each exploration attempt} lasts at most $e_i=2^{i}$ rounds), using a UXS. We say that the agent is {\em active}
during the busy part of the $\lambda$-th stage of each phase $i\geq q=\lceil\log(\lambda+1)\rceil$.  In order to explore the graph, the agent keeps estimates of the values of $c$ and $P(n)$. (Recall that the latter is the length of a UXS that allows to traverse all edges of any graph of size at most $n$, starting from any node). The values of these estimates in phase $i$ are  called $c_i$ and $u_i$, respectively, and grow depending on the strategy of the adversary.  For the first phase $q$ in which the agent is active, we set $u_q=1$ and $c_q=2^{q}$. In phase $i$ the agent uses the UXS of length $u_i$.
Call this sequence $S$. The agent uses this UXS proceeding by steps. Steps correspond to terms of the sequence $S$. 
During phase $i$, the $k$-th step consists of $s_i$ rounds during which the agent tries to move, using port $(p+S[k] \mod {d})$ (where $d$ is the degree of the current node and $p$ is the port by which the agent entered the current node), until it succeeds or until the $s_i$ rounds of the $k$-th step are over. If it succeeded to move, it waits until the $s_i$ rounds of the step are over. If the agent succeeds to perform all of its three UXS explorations during a phase $i$, i.e., if it succeeds to move once in each step, then we set $u_{i+1}=2u_i$ and $c_{i+1}=c_i$. Otherwise, we set $u_{i+1}=u_i$ and $c_{i+1}=2c_i$. {When the agent is not active, it waits at its current node}. The agent executes this algorithm until it meets the other agent.

Below we give the pseudocode of the algorithm that works for an agent with label $\lambda$
starting at an arbitrary node of any connected graph. 

\smallskip

\noindent
{\bf Algorithm Graph-RV-BF}

\smallskip

\noindent
$q:=\lceil\log(\lambda+1)\rceil$;\\
$c_q:=2^q$;\\
$u_q:=1$;\\
$i:=0$;\\
{\bf while} rendezvous not achieved {\bf do}\\
\hspace*{1cm}{\bf for}	 $j:=0$ {\bf to} $2^i-1$ {\bf  do}\\
\hspace*{2cm}{\bf if} $\lambda=j$ {\bf then}\\
\hspace*{3cm}{\em success} $:=$ {\tt true};\\
\hspace*{3cm}{\bf for} $r:=0$ {\bf to} 2 {\bf  do}\\
\hspace*{4cm}{\em success} $:=$ ({\em success} AND {\tt exploration}$(u_i,c_i)$);\\	
\hspace*{4cm}/*the value of {\tt exploration}$(u_i,c_i)$) may be different\\ 
\hspace*{4cm}in different iterations of the loop, due to the actions\\ 
\hspace*{4cm}of the adversary*/\\
{\hspace*{3cm}{\bf endfor}}\\		
\hspace*{3cm}{\bf if} {\em success} {\bf then}\\
\hspace*{4cm}$u_{i+1}:=2u_i$;  $c_{i+1}:=c_i$;\\
\hspace*{3cm}{\bf else}\\
\hspace*{4cm}$u_{i+1}:=u_i$;  $c_{i+1}:=2c_i$;\\
{\hspace*{3cm}{\bf endif}}\\		
\hspace*{3cm}{\tt wait} for $13\cdot 2^i$ rounds;\\
\hspace*{2cm}{\bf else}\\
\hspace*{3cm}{\tt wait} for $2^{i+4}$ rounds;\\
{\hspace*{2cm}{\bf endif}}\\
{\hspace*{1cm}{\bf endfor}}\\
\hspace*{1cm}$i:=i+1$;\\
{{\bf endwhile}}
\hfill $\diamond$

\smallskip

We use the following function for exploration.

\smallskip

\noindent
boolean {\tt exploration}(integer $u$, integer $c$)\\
{\em success} $:=$ {\tt true};\\
$S:=$ UXS of length $u$;\\	
{\bf for} $k:=0$ {\bf to} $u-1$ {\bf  do}\\
\hspace*{1cm}{\em moved} $:=$ {\tt false};\\
\hspace*{1cm}{\bf for} $r:=0$ {\bf to} $c-1$ {\bf  do}\\
\hspace*{2cm}$p:=$ port by which the agent entered the current node;\\
\hspace*{2cm}$d:=$ degree of the current node;\\
\hspace*{2cm}try to move by port $(p+S[k] \mod {d})$;\\
\hspace*{2cm}{\bf if} move is successful {\bf then}\\
\hspace*{3cm}{\em moved} $:=$ {\tt true}; break;\\
{\hspace*{2cm}{\bf endif}}\\
{\hspace*{1cm}{\bf endfor}}\\
\hspace*{1cm}{\tt wait} for $c-1-r$ rounds;\\
\hspace*{1cm}{\bf if} {\em moved} $=$ {\tt false} {\bf then}\\
\hspace*{2cm}{\em success} $:=$ {\tt false}; break;\\
{\hspace*{1cm}{\bf endif}}\\
{{\bf endfor}}\\
{\tt wait} until the call of {\tt exploration} lasts for $2^{cu}$ rounds;\\
return {\em success};
\hfill $\diamond$

\begin{theorem} \label{th:algocost}  
Algorithm {\tt Graph-RV-BF} is a correct rendezvous algorithm with bounded adversarial faults in arbitrary graphs, and works at cost polynomial in the size $n$ of the graph, and logarithmic
in the fault bound $c$ and  in the larger label $L$.
\end{theorem}

\begin{proof}
Let $a$ with label $\lambda$ be the first agent to be activated by the adversary. The second agent $a'$ with label $\lambda'$ is activated after $\delta\geq 0$ rounds. Let $S_i$ (respectively $S_i'$) be the $\lambda$-th (respectively $\lambda'$-th) stage of phase $i$ of execution of agent $a$ (respectively agent $a'$). The stage $S_i$ (respectively $S_i'$) is called the \emph{active} stage of agent $a$ (respectively $a'$) for phase $i$. Recall that agent $a$ (respectively agent $a'$) is \emph{active} in round $t$ if it is executing a busy part of an active stage $S_i$ (respectively $S_i'$). Let $t_k$ (respectively $t_k'$) be the round in which stage $S_k$ (respectively $S_k'$) starts. We say that stage $S_i$ of agent $a$ intersects stage $S_j'$ of agent $a'$ if there is a round during the execution when both agents are active, with agent $a$ executing {the busy part of} stage $S_i$ and agent $a'$ executing {the busy part of} $S_j'$. We say that an active stage of an agent is \emph{useless} if it intersects an active stage of the other agent. We denote by $u_i'$ and $c_i'$ the values of the variable $u_i$ and $c_i$ during the execution of Algorithm {\tt Graph-RV-BF} by agent $a'$. For simplicity, we denote these values for agent $a$ by $u_i$ and $c_i$. First, we show the following claim.

\noindent
{\bf Claim 1.}
At most three active stages of agent $a$ are useless.

In order to prove the claim,
we can assume that agent $a$ has at least one useless stage, otherwise the claim is proved. Let $S_i$ be the first useless stage of $a$. We consider the minimal index $j>i+1$ such that the active stage $S_j$ is useless. If $j$ does not exist then the claim is proved, since only stages $S_i$ and $S_{i+1}$ can be useless. We will show that for $r>j$, stage $S_r$ is not useless. This will prove the claim since in this case only the stages $S_i$, $S_{i+1}$ and $S_{j}$ can be useless. Observe that $\delta<\sum_{k=0}^{i}p_k$, since the second agent was active during the phase $i$ of the first agent. For all $k\in\mathbb{N}$, we have $p_{k+1}=4p_{k}$. Hence, we have $\sum_{k=0}^{i}p_k< \frac{1}{2}p_{i+1}$. Thus {$\delta< \frac{1}{2}p_{k-1}$} for any $k> i+1$. Hence, for any $k>i+1$ the phase $k$ of agent $a$ happens {between the beginning of} the second half of phase $k-1$ and {the end of} phase $k$ of agent $a'$. Observe that agent $a'$ was active during some phase $i'\leq i$ since agent $a$ has a useless stage in phase $i$. Hence $\lambda'\leq 2^i\leq \frac{1}{2} 2^{k-1}$ for any $k\geq i+1$. 
{It implies that, for all $k\geq i+1$, stage $S_{k-1}'$ happens during the first half of phase $k-1$ of agent $a'$, and so before stage $S_k$. For any $k,r$ such that $k\geq i+1$ and $r<k$, stage $S_{r}'$ cannot intersect stage $S_k$.}
 Hence,  for any $k> i+1$, $S_k$ can only intersect stage $S_k'$. In particular, this implies that during phase $j$, the active stage $S_j$ intersects stage $S_j'$. Hence $\lambda'<\lambda$. We have $t_j'< t_j+b_j$ since $S_j'$ intersects $S_j$ {(recall that $S_j'$ intersects $S_j$ if there exists a round in which agents $a$ and $a'$ are both executing the busy part of  $S_j$ and $S'_j$, respectively)}. Let $k\geq j$. We have {$t_{k+1}=t_k+(2^k-\lambda)s_k+\lambda s_{k+1}=t_k+p_k+\lambda(s_{k+1}-s_k)$}. Similarly, we have $t_{k+1}'=t_k'+p_k+\lambda'(s_{k+1}-s_k)$ for any $k\geq j$. We obtain that for any $k\geq j$:
\begin{eqnarray*}
 t_{k+1}-t_{k+1}'&=&t_k+p_k+\lambda(s_{k+1}-s_k)-\left(t_k'+p_k+\lambda'(s_{k+1}-s_k)\right)\\
 &=&t_k-t_k'+(\lambda-\lambda')(s_{k+1}-s_k)\\
 &\geq & t_k-t_k'+(s_{k+1}-s_k),\quad\quad\mbox{since $\lambda>\lambda'$.}
\end{eqnarray*}

By induction on $r$, we have for any $r\geq j$ that :
\begin{eqnarray*}
 t_{r+1}-t_{r+1}'&\geq &t_j-t_j'+\sum_{k=j}^{r}(s_{k+1}-s_k)\\
&\geq& t_j-t_j'+s_{r+1}-s_j\\
&\geq& t_j-t_j'+\frac{1}{2}s_{r+1}\quad\mbox{(since $s_j\leq \frac{1}{2}s_{r+1}$)}\\
&\geq& \frac{1}{2}s_{r+1}-b_j\quad\mbox{(since $t_j'< t_j+b_j$)}\\
&\geq& 2b_{r+1}-b_j\quad\mbox{(since $s_{r+1}>4b_{r+1}$)}\\
&>& b_{r+1}\quad\mbox{(since $b_{r+1}> b_{j}$).}\\
\end{eqnarray*}
For any $r\geq j$, we conclude that the {busy part of} stage $S_{r+1}'$ of agent $a'$ starts in round $t_{r+1}'$ and ends in round $t_{r+1}'+b_{r+1}$. Stage $S_{r+1}'$ does not intersect $S_{r+1}$ since $t_{r+1}>t_{r+1}'+b_{r+1}$. Hence, for any $r>j$, stage $S_r$ is not useless. This ends the proof of the claim.

In order to bound the number of moves of the two agents, we will show the following claim.

\noindent
{\bf Claim 2.}
There exists a constant $d$, such that, for any integer $i$, either rendezvous occurs by the end of the execution of phase $i$ by one of the agents, or the values of $u_i$ and $u_i'$ are at most $dP(n)$. 

In order to prove the claim, 
first consider the values of $u_i$ (for agent $a$). Let $j=\min\{k\in\mathbb{N}\mid u_k\geq P(n)\}$. 
For $i<j$ we have $u_i \leq u_j\leq 2P(n)$.
Let $i\geq j$. Consider two cases: either (1) $S_i$ is not useless or (2) $S_i$ is useless. For Case (1), consider two subcases: either (1.1) the agent succeeds to move in each step of $S_i$, or (1.2) the agent is blocked by the adversary in some step of $S_i$. In case (1.1), agent $a$ explores all the graph since it performs a UXS of length $u_k$ for $u_k\geq P(n)$. Hence, rendezvous occurs since agent $a'$ does not move during stage $S_i$. For Case (1.2), we have $u_{i+1}=u_i$.  Since Case (2) can only happen three times by Claim 1, {in view of Case (1.2) and of the fact that $u_{z+1}\leq 2u_z$ for all $z\geq0$}, we have that for all $i$, $u_i\leq 8u_j \leq 16P(n)$. This proves the claim for agent $a$ with $d=16$.

Next, we consider the values of $u'_i$ (for agent $a'$). Let $x=\min\{k\in\mathbb{N}\mid u_k'\geq P(n)\}$. 
For $i<x$ we have $u'_i \leq u_x'\leq 2P(n)$. Let $i\geq x$. 
Consider two cases: either (1) $S_i'$ is not useless or (2) $S_i'$ is useless. Using the same argument as for agent $a$, we can assume that the value {$u'_{i+1}$} is equal to {$u'_i$} in Case~(1), {otherwise rendezvous occurs in phase $i$ of agent $a'$}. Consider Case~(2). We have that $S_i'$ intersects $S_j$ for some $j\geq i$. Let $p=\min\{k\in\mathbb{N}\mid \mbox{$S_k'$ intersects $S_j$}\}$ and $q=\max\{k\in\mathbb{N}\mid \mbox{$S_k'$ intersects $S_j$}\}$. We consider two subcases: (2.1) $q-p>0$ or (2.2) $q-p=0$. Consider Subcase~(2.1). Stages $S_{q-1}'$ and $S_q'$ intersect stage $S_j$ and we have $t_j\leq t_{q-1}'+b_{q-1}\leq t_q' \leq t_j+b_j$. This implies $t_{q}'-(t_{q-1}'+b_{q-1})\leq b_j$. {We have $t_{q}'-t_{q-1}'=p_{q-1}+\lambda'(s_{q+1}-s_q)\geq p_{q-1}$.} Hence:

\begin{eqnarray*}
p_{q-1} &\leq& b_{j}+b_{q-1}\\
2^{2(q-1)+4} &\leq& 3 \cdot 2^{j}+3\cdot 2^{q-1},\quad \mbox{since $p_i=2^{2i+4}$ and $b_i=3.2^{i}$ for all $i\in\mathbb{N}$}.\\
2^{2(q-1)+4} &\leq& 3 \cdot 2^{j}+3 \cdot 2^{j-1},\quad \mbox{since $S_q'$ intersects $S_j$ and $q\leq j$}.\\
2^{2q}&\leq & \frac{9}{8} \cdot 2^{j}\\
2q&\leq& j+\log\left(\frac{9}{8}\right)\\
2q&\leq& j+0.06\\
2q&\leq& j,\quad\mbox{since both $q$ and $j$ are integers.}\\
\end{eqnarray*}

For any $p\leq k\leq q$, we have $2k\leq j$. Now we show that for $k$ such that $p\leq k\leq q$, if $u_k'\geq 48P(n)$, then rendezvous occurs. Assume that $u_k' \geq 48P(n)$. We have :
{
\begin{eqnarray*}
2^{2k}&\leq &2^j \quad \mbox{since $2k\leq j$}\\
(e_k)^2 & \leq & e_j \quad \mbox{since $\forall i\in\mathbb{N}, e_i=2^i$}\\
e_k(c_k'u_k') & \leq & c_ju_j \quad \mbox{since $\forall i\in\mathbb{N}, e_i=u_i'c_i'=u_ic_i$}\\
48P(n)e_kc_k' & \leq & 16P(n) c_j \quad \mbox{since $u_k'\geq 48P(n)$ and $u_j\leq 16P(n)$}\\
3e_k & \leq & c_j \quad \mbox{since $c_k'\geq 1$.}\\
\end{eqnarray*}}

Hence, during a period of time of duration less or equal to $c_j$, agent $a'$ explores three times the graph. Agent $a$ moves at most twice during this period of time, since this period of time intersects at most two steps and agent $a$ moves at most once in each step. Hence, agent $a'$ explores the graph at least once while agent $a$ is not moving and thus rendezvous occurs. It follows that for $k$ such that $p\leq k\leq q$, if $u_k'\geq 48P(n)$, then rendezvous occurs. Hence, for $k$ such that $p\leq k\leq q$, 
if rendezvous does not occur by the end of phase $k$, we have $u_k'\leq 96P(n)$. (By definition of $p$ and $q$, in Subcase (2.1)
it is enough to restrict attention to $k$ in the interval $[p,q]$.) 

Finally, let us consider Subcase (2.2). The stage $S_i'$ is the only active stage of $a'$ to intersect $S_j$. This means that Subcase (2.2) can only occur three times by Claim 1. Hence the value of $u_i$ is less than $\max\{2^3 u_j,2^2 96P(n)\}\leq 384P(n)$. This proves the claim for agent $a'$ with $d=384$ and concludes the proof of Claim~2.

{We show that rendezvous occurs by the end of the execution of phase $\rho(r)=\log(P(n))+\log(c)+\log(L)+r$ by agent $a$, for some constant $r$. Note that for any phase $i\geq\lceil\log(L)\rceil \geq\lceil\log(\lambda)\rceil$, agent $a$ has an active stage $S_i$. At the end of each phase, either $u_{i+1}=2u_i$ or $c_{i+1}=2c_i$. For all $i$, if rendezvous has not occurred by the end of phase $i$, we have $u_i\leq dP(n)$ for some constant $d$ (cf. Claim 2). Moreover, observe that if $c_j\geq c+1$ then $u_{j+1}=2u_j$, since the adversary cannot prevent the move of the agent more than $c$ times. Hence, there exists a constant $y$, such that if rendezvous has not occurred by the beginning of phase $\rho(y)$, then we have $u_{\rho(y)}\geq P(n)$ and $c_{\rho(y)}\geq c+1$. This means that during any active stage $S_j$ with $j\geq \rho(y)$, agent $a$ explores the graph. At least one of the stages among $S_{\rho(y)}$, $S_{\rho(y)+1}$, $S_{\rho(y)+2}$ or $S_{\rho(y)+3}$ is not useless by Claim 1. Hence, rendezvous occurs by the end of the execution of phase 
$\rho(y)+3=\rho(y+3)=\rho(r)$ by agent $a$, with $r=y+3$.}

{Now we can conclude the proof of the theorem.
Since agent $a$ is activated before or at the same time as agent $a'$, agent $a'$ cannot execute more than the first $\rho(r)$ phases before rendezvous occurs. Hence, at most $\rho(r)$ phases are executed by each of the agents $a$ and $a'$. Since, in view of Claim 2, each agent makes at most $dP(n)$ moves during each phase, the total number of moves of the agents is in $O((\log(P(n))+\log(c)+\log(L))P(n))$. Since $P$ is a polynomial, the theorem follows.}

\end{proof}

Notice that in the bounded fault scenario (as opposed to the unbounded fault scenario) it makes sense to speak about the time of a rendezvous algorithm execution
(i.e., the number of rounds from the start of the earlier agent until rendezvous), apart from its cost.
Indeed, now the time can be controlled by the algorithm. Our last result  gives an estimate on the execution time of Algorithm {\tt Graph-RV-BF}.

\begin{theorem}
Algorithm {\tt Graph-RV-BF} works in time polynomial in the size $n$ of the graph,
in the fault bound $c$ and in the larger label $L$ .  
\end{theorem}

\begin{proof}
From the proof of Theorem \ref{th:algocost} {we know that rendezvous occurs by the time when the agent activated earlier executes its $\rho$-th phase with $\rho=\log(P(n))+\log(c)+\log(L)+O(1)$}. Thus the number of rounds until rendezvous (since the activation of the first agent) is $\sum_{i=0}^{\rho}p_i$. We have:
{
\begin{eqnarray*}
\sum_{i=0}^{\rho}p_i&\leq &2p_{\rho},\quad\mbox{since $\forall i\in \mathbb{N}, p_{i+1}=4p_i$}\\
&\leq &2^{2\rho+5},\quad\mbox{since $\forall i\in \mathbb{N}, p_{i+1}=2^{2i+4}$}\\
&=& O((P(n) \cdot c \cdot L)^2).\\
\end{eqnarray*}}

Since $P$ is a polynomial, Algorithm {\tt Graph-RV-BF} works in time polynomial in the size $n$ of the graph, in the fault bound $c$ and in the larger label $L$.

\end{proof}

Notice the difference between the estimates of cost and of time of Algorithm {\tt Graph-RV-BF}: while we showed that cost is polylogarithmic in $L$ and $c$, for time we were only able to show that it
is polynomial in $L$ and $c$.  Indeed, Algorithm {\tt Graph-RV-BF} relies on a technique similar to ``coding by silence'' in the time-slice algorithm for leader election \cite{Ly}:
``most of the time'' both agents stay idle, in order to guarantee that agents rarely move simultaneously.
It remains open whether there exists a rendezvous algorithm with bounded adversarial faults, working for arbitrary graphs, whose both cost and time are
polynomial in the size $n$ of the graph, and polylogarithmic
in the fault bound $c$ and in the smaller label $\ell$ .

\section{Conclusion}

We presented algorithms for rendezvous with delay faults under various distributions of faults. Since we assumed no knowledge of any bound on
the size of the graph, for unbounded adversarial faults rendezvous is impossible, even for the class of rings. Hence it is natural to ask how the
situation changes if a polynomial upper bound $m$ on the size of the graph is known to the agents. In this case, even under the harshest
model of unbounded adversarial faults, a simple rendezvous algorithm can be given. In fact this algorithm mimics the asynchronous rendezvous algorithm (without faults) from \cite{DGKKP}. 
An agent with label $\lambda$, starting at node $v$ of a graph of size at most $m$, repeats $(P(m)+1)^\lambda$ times the trajectory $R(m,v)$, which starts and ends at node $v$, 
and stops.
Indeed, in this case, the number of integral trajectories $R(m,v)$ performed by the agent with larger label is larger than the number of edge traversals by the other agent, and consequently, if they have not met before,  the larger agent must meet the smaller one after the smaller agent stops, because the larger agent will still perform at least one entire trajectory afterwards. The drawback of this algorithm is that, while its cost is
polynomial in $m$, it is exponential in the smaller
label $\ell$. We know from Theorem \ref{lb} that the cost of any rendezvous algorithm must be at least linear in $\ell$,
even for the two-node tree. Hence an interesting
open problem is:

\begin{quotation}
Does there exist a deterministic rendezvous algorithm, working in arbitrary graphs for unbounded adversarial faults, with cost polynomial in the size of the
graph and in the smaller label, if a polynomial upper bound on the size of the graph is known to the agents?
\end{quotation}

\bibliographystyle{plain}


\end{document}